\documentclass[a4paper,twocolumn,10pt,accepted=2025-07-22]{quantumarticle}
\pdfoutput=1
\usepackage[utf8]{inputenc}
\usepackage[english]{babel}
\usepackage[T1]{fontenc}
\usepackage[numbers,sort&compress]{natbib}

\usepackage{graphicx}
\usepackage[dvipsnames]{xcolor}
\usepackage{rotating}
\usepackage{amsmath,amssymb,graphics,amsthm,isomath}
\usepackage{amsfonts,dsfont,mathtools}
\usepackage{bbm}
\usepackage{bm}
\usepackage{array}
\usepackage{appendix}
\newcolumntype{P}[1]{>{\centering\arraybackslash}p{#1}}
\usepackage{multirow}
\usepackage{physics}
\usepackage{braket}
\usepackage{adjustbox}
\usepackage{stmaryrd}

\usepackage[colorlinks=true, urlcolor=violet, linkcolor=blue, citecolor=red, hyperindex=true, linktocpage=true]{hyperref}
\usepackage[capitalise,compress]{cleveref}

\allowdisplaybreaks

\newtheorem{thm}{Theorem}
\numberwithin{thm}{section}

\newtheorem{lem}[thm]{Lemma}

\newtheorem{defn}[thm]{Definition}

\renewcommand{\thesection}{\arabic{section}}
\renewcommand{\thesubsection}{\thesection.\arabic{subsection}}

\makeatletter
\renewcommand{\p@subsection}{}
\renewcommand{\p@subsubsection}{}
\makeatother

\usepackage{xcolor}
\usepackage{mathtools}

\newcommand\bea{\begin{eqnarray}}
\newcommand\eea{\end{eqnarray}}
\newcommand\be{\begin{equation}}
\newcommand\ee{\end{equation}}
\newcommand\bes{\begin{subequations}}
\newcommand\ees{\end{subequations}}
\newcommand\bed{\begin{displaymath}}
\newcommand\eed{\end{displaymath}}
\newcommand\beal{\begin{aligned}}
\newcommand\eeal{\end{aligned}}
\newcommand\bew{\begin{widetext}}
\newcommand\eew{\end{widetext}}
\newcommand\beit{\begin{itemize}}
\newcommand\eeit{\end{itemize}}
\def\bea{\begin{array}}
\def\eea{\end{array}}
\newcommand\been{\begin{enumerate}}
\newcommand\eeen{\end{enumerate}}

\def\i{\text{i}}

\newcommand{\ident}[0]{\mathds{1}}

\usepackage{verbatim}

\newcommand{\im}{\operatorname{im}}

\newcommand{\transpose}[0]{\mathsf{T}}

\usepackage{tikz-cd}

\usepackage[ruled,lined]{algorithm2e}

\SetCommentSty{mycommfont}
\SetKwInOut{Input}{Input}
\SetKwInOut{Return}{Return}
\newcommand{\pluseq}{\mathrel{{+}{=}}}

\begin{document}

\title{Single-shot preparation of hypergraph product codes via dimension jump}

\author{Yifan Hong}
\email{yhong137@umd.edu}
\affiliation{Joint Quantum Institute \& Joint Center for Quantum Information and Computer Science, NIST/University of Maryland, College Park, MD 20742, USA}

\begin{abstract}
Quantum error correction is a fundamental primitive of fault-tolerant quantum computing. But in order for error correction to proceed, one must first prepare the codespace of the underlying error-correcting code. A popular method for encoding quantum low-density parity-check codes is transversal initialization, where one begins in a product state and measures a set of stabilizer generators. In the presence of measurement errors however, this procedure is generically not fault-tolerant, and so one typically needs to repeat the measurements many times, resulting in a deep initialization circuit. We present a protocol that prepares the codespace of constant-rate hypergraph product codes in constant depth with $O(\sqrt{n})$ spatial overhead, and we show that the protocol is robust even in the presence of measurement errors. Our construction is inspired by dimension-jumping in topological codes and leverages two properties that arise from the homological product of codes. We provide some improvements to lower the spatial overhead and discuss applications to fault-tolerant architectures.
\end{abstract}


\maketitle

\section{Introduction}

The reliability of large-scale quantum algorithms hedges on the creation and operation of fault-tolerant logical qubits. At the heart of a fault-tolerant architecture lies a quantum error-correcting code, whose role is to reverse the unwanted effect of noise introduced by the environment.

There has been immense development in fault-tolerant architectures using 2D topological codes, both on the theoretical \cite{Kitaev_2003, Bombin_2006, Dennis_2002, Horsman_2012, Fowler_2012, landahl2014, Brown_2017, zhou2024alg} and experimental fronts \cite{ryananderson2022, Google_2023, Bluvstein_2023, Google_2024}. Equipped with competitive thresholds, fault-tolerant gadgets and local implementation in planar geometries, 2D topological codes are the leading candidates for fault tolerance in near-term platforms. However, geometric locality in two spatial dimensions places fundamental limits on the physical overhead \cite{Bravyi_2009, BPT_2010} and capability \cite{Bravyi_2013_transversal} of these codes. By relaxing geometrical constraints, such as going to higher dimensions, one is able to construct codes that support more complex gadgets such as transversal non-Clifford gates \cite{Bombin_2007_3D,  Kubica_2015} and single-shot decoding \cite{Bombin_2015, Kubica_2022, Stahl_2024}. Single-shot codes are able to suppress errors with just one round of noisy syndrome measurement and are thus attractive from the perspective of decoding complexity and computational clock speeds.

Quantum low-density parity-check (LDPC) codes are a broad class of quantum stabilizer codes \cite{gottesman1997stabilizer} whose connectivity is few-body: every qubit is involved in a constant number of parity checks, and every parity check acts on a constant number of qubits \cite{Breuckmann_2021}. This class encompasses the topological codes mentioned above but also includes families of codes with asymptotically optimal parameters such as constant rate \cite{Freedman_2002, Tillich_2014_HGP} and/or linear distance \cite{Panteleev_2022_LDPC, qTanner_codes, Dinur_2023_LDPC}. These more exotic code families cannot be embedded in any finite spatial dimension without long-range interactions \cite{Baspin_2022, baspin2023improved, dai2024local}. Despite their geometrical nonlocality, certain LDPC codes show feasible promise in platforms capable of long-range connectivity such as those based on trapped ions \cite{Cirac_1995, chen2023ionq, Quantinuum_H2-1} and neutral atoms \cite{Saffman_2016, Jenkins_2022, Evered_2023, Bluvstein_2023}, with recent experimental demonstrations showcasing their performance \cite{hong2024ghz, berthusen2024ss, reichardt2024tess}. In this work, we will focus on one particular code family called hypergraph product (HGP) codes \cite{Tillich_2014_HGP}. Various HGP code instances can achieve constant rate (but sublinear distance), single-shot decoding \cite{Fawzi_2018}, and passive error correction \cite{hong2024thermal}. HGP codes also possess distance-preserving stabilizer measurements \cite{manes2023hgp} as well as an arsenal of fault-tolerant gadgets \cite{Krishna_2021, Breuckmann_2024_fold, Quintavalle_2023, LRESC} that can enable fault-tolerant quantum computation with constant overhead \cite{gottesman2014, Fawzi_2018, Xu_2024_constant}.

Although constant-rate HGP codes support single-shot and passive error correction, these features crucially rely on the assumption that the codespace has already been prepared. The standard procedure for preparing these HGP codes from an unentangled product state is to perform $\mathrm{\Theta}(d) = \mathrm{\Theta}(\sqrt{n})$ rounds of repeated syndrome measurement in order to reliably obtain an initial syndrome baseline or Pauli frame \cite{Dennis_2002}. Since HGP codes can be decoded in constant depth after initialization using single-shot decoding, this initial $\mathrm{\Theta}(d)$ depth presents a temporal bottleneck in their implementation.

In this work, we introduce a single-shot protocol that prepares the codespace of constant-rate HGP codes in $O(1)$ depth and $O(d)$ spatial overhead; see Fig. \ref{fig:protocol sketch} for a schematic. Our construction generalizes dimension-jumping \cite{Bombin_2016}, a specific form of code-switching, from topological codes to more general quantum LDPC codes. 2D $\leftrightarrow$ 3D dimension-jumping in the surface code has been previously studied as a way to perform non-Clifford gates \cite{Brown_2020, Scruby_2022_jit} and generate long-range localizable entanglement \cite{3D_cluster, Bravyi_2020_sc}, but doing so in a single shot sacrifices performance \cite{Scruby_2022}. In contrast, since HGP codes support single-shot error correction, we can easily show that single-shot dimension-jumping is robust, and we can furthermore use existing single-shot HGP decoders \cite{Fawzi_2018, Higgott_2023} to correct errors during the process.

\begin{figure}[t]
    \centering
    \includegraphics[width=0.4\textwidth]{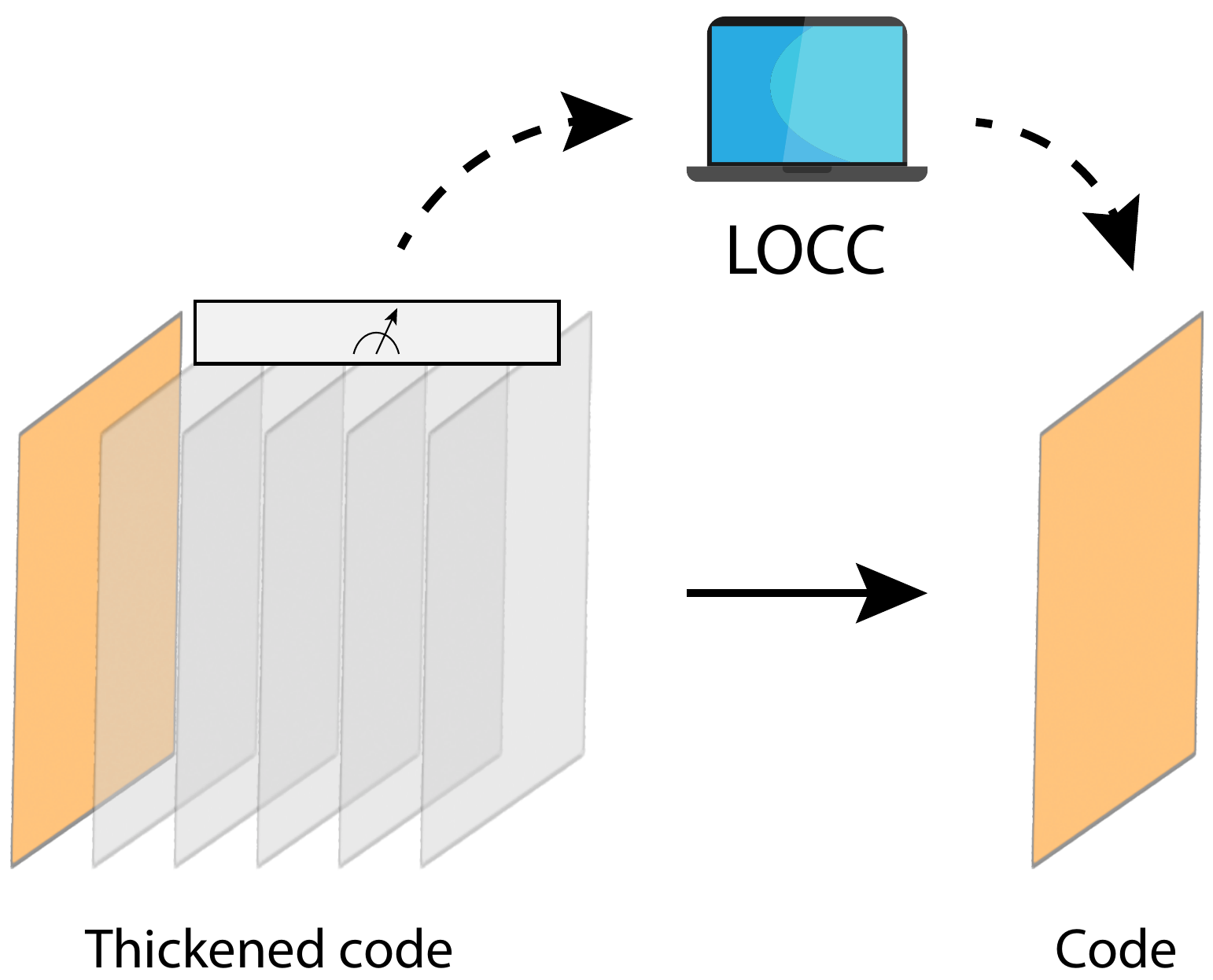}
    \caption{A high-level schematic of the state-preparation protocol is illustrated. The protocol involves fault-tolerant code-switching between the target HGP code and a ``thickened'' code with enhanced error-correcting properties. Local operations and classical communication (LOCC) in the form of measurements and feedback will be used to correct errors.}
    \label{fig:protocol sketch}
\end{figure}

The paper is organized as follows. In Section \ref{sec:background}, we review some preliminary mathematical background regarding HGP codes and the necessary algebraic machinery needed to describe our protocol. In Section \ref{sec:protocol}, we introduce the state-preparation protocol, provide a decoding algorithm, and analyze its fault tolerance against adversarial noise. In Section \ref{sec:beyond rep}, we discuss generalizations and introduce a modification that reduces the spatial overhead. In Section \ref{sec:numerics}, we conduct numerical simulations of the protocol under an independent noise model. In Section \ref{sec:outlook}, we provide some concluding remarks and discuss applications to fault-tolerant architectures based on HGP codes. The appendices provide formal statements and proofs.


\section{Background}\label{sec:background}

\subsection{Hypergraph product codes}

In this section, we briefly review the construction of quantum codes from classical codes via the hypergraph product and introduce some necessary machinery to describe the state preparation protocol.

Given $n$ bits, each taking on a value of 0 or 1, the set of all $2^n$ possible bitstrings forms an $n$-dimensional binary vector space $\mathbb{F}^n_2$, where addition is performed modulo 2. A classical $[n,k,d]$ linear code encoding $k$ logical bits is formally a $k$-dimensional subspace of logical codewords within this $n$-dimensional space. We can choose $k$ basis vectors in this subspace to define a $k \times n$ generator matrix $G \in \mathbb{F}^{k\times n}_2$, which is a compact representation of the code. The code distance $d$ is defined as the minimum number of bit flips to transition between any two codewords. Since in a linear code, the difference between two codewords is also a codeword, the code distance is simply the minimum Hamming weight amongst all $2^k-1$ nonzero codewords. It is often useful to alternatively present this linear code using an $m \times n$ parity-check matrix $H \in \mathbb{F}^{m\times n}_2$ which satisfies $HG^\transpose = 0$ and so $\operatorname{rs}(G) \subseteq \ker(H)$; in other words, $H$ encodes parity constraints that every codeword must satisfy. Note that there exist many different presentations of $H$, similar to our freedom in choosing basis vectors for $G$. We say that $H$ is $(\Delta_{\rm c}, \Delta_{\rm r})$-LDPC if the nonzero weight of every column and row is bounded by constants $\Delta_{\rm c}, \Delta_{\rm r} = O(1)$ respectively \cite{Gallager_1962}. From an algebraic perspective, a classical linear code can be described by the 2-term chain, or 1-complex
\begin{align}\label{eq:2-term chain}
    B \xlongrightarrow{H} S \, ,
\end{align}
where $B$ is the space of physical bitstrings or errors, and $S$ is the space of violated parity checks or error syndromes \cite{Bombin_2007_Hom}.

Similar to the classical story, a quantum $\llbracket n,k,d \rrbracket$ stabilizer code is a $2^k$-dimensional subspace of the $2^n$-dimensional Hilbert space on $n$ qubits \cite{gottesman1997stabilizer}. In analogy to the classical parity constraints, the quantum codespace is defined as the simultaneous +1 eigenspace of a set of commuting Pauli strings (i.e. tensor products of $I,X,Y,Z$), which generates the code's stabilizer group. The code distance is similarly defined as the minimum number of local operators needed to transition between different codewords. A Calderbank-Shor-Steane (CSS) code is a special type of stabilizer code whose Pauli checks are purely $X$-type or $Z$-type \cite{Calderbank_1996, Steane_1996}. As such, they can be packaged into two binary parity-check matrices $H_X$ and $H_Z$, where a nonzero entry in a row of $H_X$ ($H_Z$) denotes the location of a Pauli $X$ ($Z$) operator. The commutativity of the Pauli checks then translates into the orthogonality condition $H^{}_X H^\transpose_Z = 0$ between the two parity-check matrices. We say a CSS code is $(\Delta_{\rm c},\Delta_{\rm r})$-LDPC if both $X$-type and $Z$-type parity-check matrices are individually $(\Delta_{\rm c},\Delta_{\rm r})$-LDPC. From an algebraic standpoint, a CSS code can be described by the 3-term chain
\begin{align}\label{eq:CSS 3-term chain}
    S_X \xlongrightarrow{H^\transpose_X} Q \xlongrightarrow{H_Z} S_Z \, ,
\end{align}
where $S_X$ ($S_Z$) is the space of error syndromes according to $H_X$ ($H_Z$), $Q$ is the space of physical qubit errors, and the orthogonality condition is explicitly expressed in terms of the composition of consecutive maps being zero in a chain complex. Logical $\bar{X}$ operators are classified according to the first homology group $\mathcal{H}_1 \equiv \ker H_Z / \im H^\transpose_X$, which consists of Pauli $X$ strings that satisfy all $Z$-checks modded out by the $X$-type stabilizer subgroup. Logical $\bar{Z}$ operators are classified according to the first cohomology group $\mathcal{H}^1 \equiv \ker H_X / \im H^\transpose_Z$.

Given two classical linear codes with parity-check matrices $H_1$ and $H_2$, we can take a (homological) tensor product \cite{Freedman_2014, bravyi2013hom} of their associated 2-term chains \eqref{eq:2-term chain} to obtain the 2-dimensional complex
\begin{equation}
\begin{tikzcd}
    & {B_1 \otimes S_2} && {S_Z} \\
    {B_1 \otimes B_2} && {S_1 \otimes S_2} & Q \\
    & {S_1 \otimes B_2} && {S_X}
    \arrow["{\ident \otimes H_2}", from=2-1, to=1-2]
    \arrow["{H^\transpose_1 \otimes \ident}"', from=2-3, to=1-2]
    \arrow["{H_Z}", from=2-4, to=1-4]
    \arrow["{H^\transpose_1 \otimes \ident}", from=3-2, to=2-1]
    \arrow["{\ident \otimes H_2}"', from=3-2, to=2-3]
    \arrow["{H^\transpose_X}", from=3-4, to=2-4]
\end{tikzcd}
\end{equation}
which we can then ``collapse'' into a one-dimensional 3-term chain by combining the spaces $B_1 \otimes B_2$ and $S_1 \otimes S_2$. The HGP code is then defined as the quantum CSS code associated with this 3-term chain \cite{Tillich_2014_HGP}. Explicitly, the CSS parity-check matrices take the form
\begin{subequations}\label{eq:HGP parameters}
\begin{align}
    H^{}_X &= \big(\, H^{}_1 \otimes \ident_{n_2} \;\big|\; \ident_{m_1} \otimes H^\transpose_2 \,\big) \\
    H^{}_Z &= \big(\, \ident_{n_1} \otimes H^{}_2 \;\big|\; H^\transpose_1 \otimes \ident_{m_2} \,\big) \, ,  \label{eq:HGP H_Z}
\end{align}
\end{subequations}
and we can verify that $H^{}_X H^\transpose_Z = 2\big( H^{}_1 \otimes H^\transpose_2 \big) = 0$ over $\mathbb{F}_2$. Geometrically, the tensor products in \eqref{eq:HGP parameters} bestow a certain ``two-dimensional'' structure to the HGP code, where the connectivity of one classical code is ``copied'' along one dimension and the other along the second dimension; see Fig. \ref{fig:2D HGP layout} for an illustration. This structure makes HGP codes especially appealing for hardware which can connect distant rows and columns of qubits in parallel such as the neutral-atom platform \cite{Xu_2024_constant}. Denoting transposed code parameters with the transpose symbol, the $\llbracket n,k,d \rrbracket$ parameters of the HGP code are given by
\begin{subequations}
\begin{align}
    n &= n_1n_2 + m_1m_2  \\
    k &= k_1k_2 + k^\transpose_1 k^\transpose_2  \\
    d &= \min\left( d_1, d_2, d^\transpose_1, d^\transpose_2 \right) \, ,
\end{align}
\end{subequations}
which can be obtained from properties of the homological tensor product. Many properties of a hypergraph product are inherited from its input classical codes. For example, if the input codes are LDPC, then so is the HGP code. If the input codes have constant rate (i.e. $k=\mathrm{\Theta}(n)$), then the HGP code will as well. However, the code distance can only achieve an asymptotic scaling of $d = O(\sqrt{n})$.

\begin{figure}[t]
    \centering
    \includegraphics[width=0.35\textwidth]{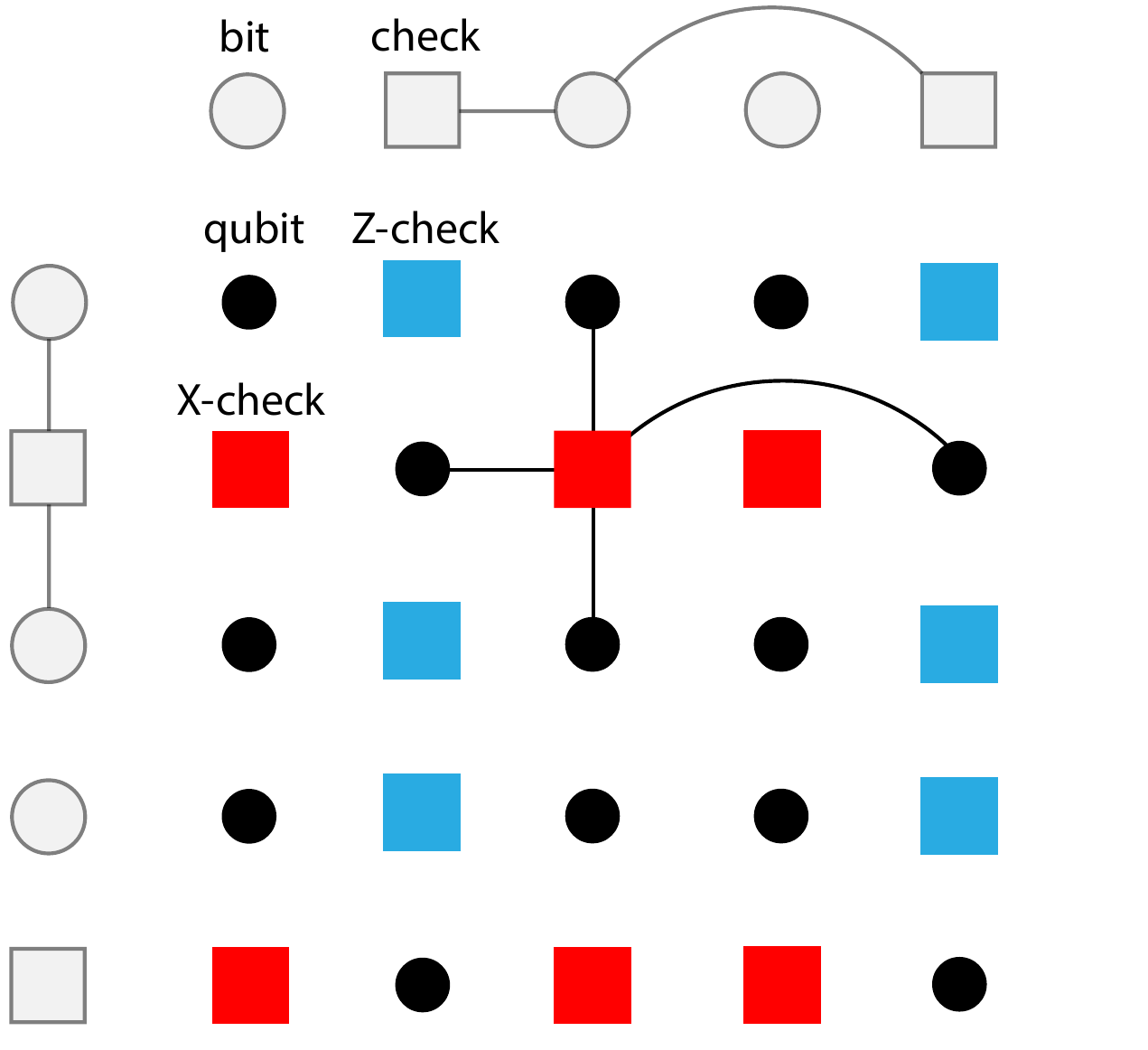}
    \caption{The hypergraph product of two classical linear codes is depicted in a 2D layout. Gray circles and squares denote classical bits and parity checks respectively. Solid black circles, red squares, and blue squares denote qubits, $X$-checks, and $Z$-checks respectively. The connectivity of the resulting HGP code is inherited from its classical codes, as shown for a subset of edges.}
    \label{fig:2D HGP layout}
\end{figure}

Since the hypergraph product works for any two classical codes, optimizing the quantum code parameters reduces to a problem of optimizing the classical code parameters. Currently, the best classical code parameters are achieved by probabilistic methods. For instance, random classical LDPC codes are known to approach the Gilbert–Varshamov bound, a general rate-distance tradeoff, for linear codes \cite{Mosheiff_2021}.


\subsection{Single-shot error correction}

For stabilizer codes, error correction typically proceeds by (\emph{i}) measuring the set of Pauli check operators to obtain the error syndrome and then (\emph{ii}) decoding the error syndrome to obtain a correction operator to be applied back to the code. The measurement step will collapse any arbitrary error into a Pauli string, a phenomenon known as error digitization, and so it suffices to focus on only decoding Pauli errors \cite{Shor_1995}. For CSS codes, the analysis further factorizes into independently correcting $X$-type and $Z$-type errors. Since our focus is on HGP codes, which are CSS codes, we will hence drop the $X/Z$ subscript for brevity unless otherwise noted.

Given a physical error $\mathbf{e} \in Q$, its error syndrome is defined as $\mathbf{s}(\mathbf{e}) \equiv H \mathbf{e} \in S$, where $H$ is either $H_X$ or $H_Z$. The primary objective of (maximum likelihood) decoding can be summarized as: given a syndrome $\mathbf{s}$, compute the most probable error $\hat{\mathbf{e}}$. Note that there can be many different errors, differing by elements of $\ker(H)$, that all give the exact same syndrome. In order to know which error in this set is the most probable, we need some information about the structure of the underlying noise. A popular noise model to test worst-case scenarios is adversarial noise, where one examines errors below some cutoff weight, which typically increases with the code distance. In this simple error model, the most probable error is usually the one with the lowest weight, and so one usually designs a decoder which tries to output a minimal-weight correction. Decoding is successful if the combination $\mathbf{e} + \hat{\mathbf{e}}$ leaves the codespace unchanged, i.e. it is an element of the code's stabilizer group. In this paper, we will focus on the adversarial noise model. In the outlook (Sec. \ref{sec:outlook}), we comment on the more realistic scenario of local stochastic noise.

In real hardware, syndrome extraction is a physical operation and will be subject to errors itself. As a consequence, the obtained error syndrome will not always be faithful, and if we naively decode the corrupted syndrome, our inferred correction may unintentionally introduce even more errors than before. As a example, suppose we have a large surface code with no errors on the data qubits but a single ancilla error causing a flipped check in the bulk. If we input this weight-1 syndrome into a matching decoder, it will match the flipped check to a boundary, resulting in an unintentional, residual error of potentially $O(d)$ weight that may overwhelm the code in the next QEC cycle. The typical approach to counteract measurement errors in the surface code is to perform $\mathrm{\Theta}(d)$ rounds of syndrome measurements and then decode over this entire history of syndromes \cite{Dennis_2002}. Remarkably, there exist families of codes and decoders where one can keep residual errors under control with only a single round of noisy syndrome measurement; we call such codes single-shot codes \cite{Bombin_2015}. These codes have additional structure that allows a decoder to detect and correct for measurement errors.

We now briefly review two conditions which are sufficient for single-shot error correction against adversarial noise \cite{Campbell_2019, Quintavalle_2021}, and they encompass all known single-shot codes to date. These two conditions share some similarity, but their decoder implementations can be quite different, as we will explain below.

The first condition comes from a low-error property of the parity-check matrix called confinement, and it roughly stipulates that the number of violated checks, or energy penalty, grows with the error weight, up to some cutoff. For notation, we use $\abs{\cdot}$ to denote the Hamming weight. Since in a stabilizer code, the weight of an error can be arbitrarily increased by appending stabilizer elements, we also define a reduced weight $\norm{\cdot}$, which is the minimum weight of an error over its stabilizer group orbit. Formally, confinement is defined as follows.

\begin{defn}[Confinement \cite{Quintavalle_2021}]
\label{defn:confinement}
    Let $t>0$ be an integer and $f: \mathbb{Z}\rightarrow\mathbb{R}$ an increasing function. For a parity-check matrix $H \in \mathbb{F}^{m\times n}_2$, we say it is $(t,f)$-confined if for any error $\mathbf{e} \in \mathbb{F}^{n}_2$ with reduced weight $\norm{\mathbf{e}} \leq t$, its syndrome $\mathbf{s}(\mathbf{e}) = H\mathbf{e}$ obeys
    \begin{align}
        f\big(\abs{\mathbf{s}(\mathbf{e})}\big) \geq \norm{\mathbf{e}} \, .
    \end{align}
\end{defn}

If the confinement function $f$ is an increasing function independent of $n$, and the confinement cutoff $t$ grows polynomially with $n$, then the code is single-shot against adversarial noise \cite{Quintavalle_2021}. For confined codes, decoding typically proceeds by including the parity checks themselves as noisy variables during the weight minimization. Confinement then guarantees that the low-energy state space is partitioned into clusters that are well-separated in Hamming distance \cite{Anshu_2023}. This cluster separation allows a minimum-weight decoder to stay within its starting cluster, i.e. not fail too badly, and thus any residual error will be bounded by the size of the cluster (governed by $t$ and $f$).

The second condition for single-shot error correction comes from a low-syndrome property called soundness, and it stipulates that low-weight syndromes can always be produced by low-weight errors. Formally, it is defined as follows.

\begin{defn}[Soundness \cite{Campbell_2019}]
\label{defn:soundness}
    Let $t>0$ be an integer and $f: \mathbb{Z}\rightarrow\mathbb{R}$ an increasing function. For a parity-check matrix $H \in \mathbb{F}^{m\times n}_2$, we say it is $(t,f)$-sound if for any valid syndrome $\mathbf{s} \in \im(H)$ with $\abs{\mathbf{s}} \leq t$, there exists an error $\mathbf{e} \in \mathbb{F}^n_2$ satisfying
    \begin{align}
        \norm{\mathbf{e}} \leq f\big(\abs{\mathbf{s}(\mathbf{e})}\big) \, .
    \end{align}
\end{defn}

Similarly to confinement, if the soundness function $f$ is an increasing function independent of $n$, and the soundness cutoff $t$ grows polynomially with $n$, then the code is single-shot against adversarial noise \cite{Campbell_2019}. In fact, for LDPC codes, soundness (Def. \ref{defn:soundness}) implies confinement (Def. \ref{defn:confinement}), with the two cutoffs being related by the $O(1)$ column weight of $H$ \cite{Quintavalle_2021}. Nonetheless, soundness is a stricter condition because it entails redundancies amongst the check operators \cite{Sasson_2009}, which is not a requirement for confinement. The existence of linearly dependent parity checks means that $H$ is rank deficient and has nontrivial left null vectors. In particular, we can collect these left null vectors into rows of a ``metacheck'' matrix $M$ which then satisfies $MH = 0$. From the algebraic perspective, we can then extend the original 2-term chain (\eqref{eq:2-term chain} or either side of \eqref{eq:CSS 3-term chain}) to
\begin{align}\label{eq:LTC 3-term chain}
    B \xlongrightarrow{H} S \xlongrightarrow{M} R \, ,
\end{align}
where $R$ is the space of check redundancies amongst the rows of $H$. Since soundness implies confinement in LDPC codes, we can employ the same decoding strategy as previously mentioned. However, the structure \eqref{eq:LTC 3-term chain} enables an alternative decoding strategy. Notice that $\im(H) \subseteq \ker(M)$, i.e. valid syndromes are codewords of the code associated with the $S\rightarrow R$ chain. So given a corrupted syndrome $\mathbf{s}$, we can construct a ``metasyndrome'' $\mathbf{m}(\mathbf{s}) = M\mathbf{s} \in R$. We can then perform the following two-stage decoding procedure:
\begin{enumerate}
    \item Decode the metasyndrome $\mathbf{m}$ using $M$ to obtain a repaired syndrome $\hat{\mathbf{s}}$.
    \item Decode the repaired syndrome $\hat{\mathbf{s}}$ using $H$ to obtain a correction $\hat{\mathbf{e}}$.
\end{enumerate}
The soundness property (Def. \ref{defn:soundness}) then guarantees that the residual error $\mathbf{e} + \hat{\mathbf{e}}$ has bounded size \cite{Campbell_2019}.

Examples of single-shot stabilizer codes include the 4D surface code \cite{Dennis_2002}, whose CSS parity checks possess the soundness property, and constant-rate HGP codes \cite{Leverrier_2015}, whose CSS parity checks possess the confinement property.


\section{Single-shot codespace preparation}\label{sec:protocol}

In the previous section, we reviewed two sufficient conditions for single-shot error correction as well as two decoding strategies that leverage their corresponding properties. An important assumption that went into those results was that we started in the codespace of our error-correcting code. The reason is that the notion of ``small'' for errors and syndromes is only well-defined with respect to the codespace (no errors and zero syndrome), and so in order for those results to hold, we need to assume our codespace has been reliably prepared. In practice, the codespace needs to be prepared from an unentangled product state, and designing such a protocol to be fault-tolerant can be nontrivial.

For large-block CSS codes, a popular method of preparing $\ket{\overline{\mathbf{0}}}$ or $\ket{\overline{\bm{+}}}$ codeword states is transversal initialization \cite{Dennis_2002}: to prepare logical $\ket{\overline{\bm{+}}}$,
\begin{enumerate}
    \item Initialize all physical qubits in $\ket{+}^{\otimes n}$.
    \item Measure all $Z$-type check operators to obtain an initial $Z$-syndrome.
    \item Decode the $Z$-syndrome and apply the resulting $X$-correction.
\end{enumerate}
The procedure for logical $\ket{\overline{\mathbf{0}}}$ follows upon switching the roles of $X$ and $Z$. In the absence of errors, the first step guarantees that we are in the simultaneous +1 eigenstate of all $X$-checks and logical $\bar{X}$ operators. Since $Z$-checks commute with $X$-checks, the second step maintains the first property while projecting our product state into random\footnote{up to metacheck constraints} $\pm 1$ eigenstates of the $Z$-checks. At this point, our state is effectively $\ket{\overline{\bm{+}}}$ up to a (potentially large) $Z$ error. The third and last step returns us to the simultaneous +1 eigenstate of all the $Z$-checks. Note that any suitable correction in the last step suffices because we are not worried about logical $\bar{X}$ errors on the $\ket{\overline{\bm{+}}}$ state. In practice, the correction in step 3 is not actually performed, but rather the $Z$-syndrome from the second step is kept in software as a baseline (i.e. Pauli frame) from which future syndromes are subtracted from.

Like single-shot error correction, the goal of single-shot state preparation is to initialize a desired quantum state in a fault-tolerant way such that residual errors are controlled. Examining the transversal initialization strategy, we see that the first step that initializes $\ket{+}^{\otimes n}$ consists of $n$ single-qubit operations and so is inherently fault-tolerant. If our code is LDPC, then the second step can be implemented with a constant-depth circuit of 2-qubit gates; the propagation of correlated errors is bounded, and hence this step is also inherently fault-tolerant. The third step is where dangerous errors can occur. Suppose we obtain an unfaithful syndrome in step 2 due to measurement noise. If our code is sound and contains metachecks, we can detect and repair these syndrome errors using the method described in the previous section. We then set our repaired syndrome as the baseline syndrome for step 3, and the soundness of the code guarantees that any residual error remains small; this procedure has been recently emphasized in \cite{xu2024fastparallel}. Indeed, for sound codes, there is not really a distinction between initialization and subsequent QEC cycles. 

On the other hand, if our code does not contain any metachecks, then every parity check is an independent constraint and hence has equal probability to be measured as either +1 or -1 during transversal initialization. As a consequence, all $2^m$ error syndromes are equally likely (for $m$ checks), and we will have no information at our disposal to verify whether our obtained syndrome is faithful. If we naively pretend that this syndrome is faithful, then we risk encountering the same issue that plagued single-shot decoding of the surface code: a single flipped check may propagate through the decoder to a large physical error. Because constant-rate HGP codes can exhibit confinement without any metachecks\footnote{Random classical LDPC codes have full-rank parity-check matrices and good parameters with high probability \cite{ModernCodingTheory}.} \cite{Leverrier_2015}, they fall into a niche class of codes which possess single-shot QEC but \emph{not} single-shot initialization.

\subsection{Overview of the protocol}

We now describe our single-shot protocol for preparing the logical $\ket{\overline{\bm{+}}}$ state of constant-rate HGP codes \cite{Leverrier_2015}, which support single-shot error correction due to confinement but are not sound due to the lack of metachecks. The protocol will consist of two stages, which we will independently describe and analyze in the following two subsections. The main idea will be to first construct a different, but related, ``thickened'' code that possesses soundness and hence can be prepared in a single shot via transversal initialization. Then we perform fault-tolerant code-switching to return to our original HGP code. The preparation of $\ket{\overline{\mathbf{0}}}$ is analogous.

We will use tildes to denote quantities related to the thickened code --- e.g. $\llbracket \tilde{n}, \tilde{k}, \tilde{d} \rrbracket$. The first stage of the protocol to prepare $\ket{\overline{\bm{+}}}$ can be summarized as:
\begin{enumerate}
    \item Initialize all physical qubits in $\ket{+}^{\otimes \tilde{n}}$.
    \item Measure all $Z$-type check operators to obtain an initial $Z$-syndrome $\mathbf{s}_0$ and $Z$-metasyndrome $\mathbf{m}_0 = \tilde{M}_Z \mathbf{s}_0$.
    \item Decode the $Z$-metasyndrome using $\tilde{M}_Z$ \eqref{eq:3D HGP M_Z} to obtain a repaired syndrome $\hat{\mathbf{s}}_0$.
    \item Decode the repaired syndrome $\hat{\mathbf{s}}_0$ using $\tilde{H}_Z$ \eqref{eq:3D HGP H_Z} to return to $\ket{\overline{\bm{+}}}_{\rm th}$, up to a small residual $X$ error.
\end{enumerate}
The second stage can be summarized as:
\begin{enumerate}
\setcounter{enumi}{4}
    \item Measure $X$ on all physical qubits except on one of the boundary HGP codes and reconstruct a bulk $X$-syndrome.
    \item Input the above bulk $X$-measurement outcomes, $X$-syndrome, and a suitable single-shot decoder into Algorithm \ref{alg:collapse correction} to obtain a final boundary $Z$-correction.
\end{enumerate}
Note that only the first two steps of stage 1 and the first step of stage 2 are physical operations and hence prone to noise. Since we begin in the product $\ket{+}^{\otimes\tilde{n}}$ state, the only troublesome noise during stage 1 are physical $Z$ errors and syndrome measurement errors. Syndrome errors are accounted for during the syndrome repair step, and physical $Z$ errors are corrected during the second stage. Since all physical operations are either transversal or involve a finite-depth circuit (e.g. for single-ancilla measurement schemes), error propagation is also bounded. Thus, the entire protocol will be fault-tolerant as long as we are able to reliably correct the errors mentioned above.

\subsection{Stage 1: Homological dimension jump}\label{sec:stage 1}

In the first stage, we will transversally initialize the logical $\ket{\overline{\bm{+}}}$ state of a thickened code, which is built from the homological product of our original $\llbracket n,k,d \rrbracket$ HGP code, described by a 3-term chain \eqref{eq:CSS 3-term chain}, with a suitable classical code of distance $d$, described by a 2-term chain \eqref{eq:2-term chain} \cite{Hastings_2017_weight, Evra_2020}. To make the analysis easy to follow, we will first use a 1D repetition code of length $d$ for thickening. Later in Sec. \ref{sec:beyond rep}, we will discuss alternative classical codes which go beyond repetition. Since the original HGP code can be described by a two-dimensional complex, we can interpret this step as a dimension jump from a two-dimensional complex to a three-dimensional one. Denoting $H_X, H_Z$ and $h$ as the CSS parity-check matrices of the HGP code and the parity-check matrix of the repetition code respectively, the product complex is given by
\begin{equation}\label{eq:3D HGP complex}
\begin{tikzcd}
    & {(S_Z,S)} && {R_Z} \\
    {(S_Z,B)} && {(Q,S)} & {S_Z} \\
    {(Q,B)} && {(S_X,S)} & Q \\
    & {(S_X, B)} && {S_X}
    \arrow["{\ident \otimes h}", from=2-1, to=1-2]
    \arrow["{H_Z \otimes \ident}"', from=2-3, to=1-2]
    \arrow["{\tilde{M}_Z}", from=2-4, to=1-4]
    \arrow["{H_Z \otimes \ident}", from=3-1, to=2-1]
    \arrow["{\ident \otimes h}", from=3-1, to=2-3]
    \arrow["{H^\transpose_X \otimes \ident}"', from=3-3, to=2-3]
    \arrow["{\tilde{H}_Z}", from=3-4, to=2-4]
    \arrow["{H^\transpose_X \otimes \ident}", from=4-2, to=3-1]
    \arrow["{\ident \otimes h}"', from=4-2, to=3-3]
    \arrow["{\tilde{H}^\transpose_X}", from=4-4, to=3-4]
\end{tikzcd}
\end{equation}
which we can collapse into a 4-term chain as shown on the right. We thus obtain a new CSS code with parity-check matrices
\begin{subequations}\renewcommand*{\arraystretch}{1.5}
\label{eq:3D HGP H_X and H_Z}
\begin{align}
    \tilde{H}_X &= \left(\, H_X \otimes \ident \;\big|\; \ident \otimes h^\transpose \,\right)  \\
    \tilde{H}_Z &= \left(\begin{array}{c|c}
        H_Z \otimes \ident \;\; & \mathbf{0}  \\
        \ident \otimes h & \; H^\transpose_X \otimes \ident
    \end{array}\right)  \label{eq:3D HGP H_Z}
\end{align}
\end{subequations}
and $Z$-metacheck matrix
\begin{align}\label{eq:3D HGP M_Z}
    \tilde{M}_Z = \left(\, \ident \otimes h \;\big|\; H_Z \otimes \ident \,\right)  \, ,
\end{align}
from which we can verify that $\tilde{H}^{}_X \tilde{H}^\transpose_Z = H^{}_X H^\transpose_Z \otimes \ident + 2\big( H_X \otimes h^\transpose \big) = 0$ and $\tilde{M}_Z \tilde{H}_Z = 2\big( H_Z \otimes h \big) + H^{}_Z H^\transpose_X \otimes \ident = 0$ as required. Geometrically, the Tanner graph of this thickened code consists of ``sheets'' of the original HGP code interleaved with intermediate qubits and $Z$-checks; see Fig. \ref{fig:3D HGP layout} for an illustration. The sheets are labeled by the left block in \eqref{eq:3D HGP H_X and H_Z} and the intermediate qubits by the right block. For simplicity, we will assume that the two classical input codes for the HGP code are the same, so that $H_X$ and $H_Z$ both have $m$ rows. The quantum code parameters of the thickened code are then given by \cite{Zeng_2019}
\begin{subequations}
\begin{align}
    \tilde{n} &= nd + m(d-1) = O(n^{3/2})  \label{eq:tilde n}  \\
    \tilde{k} &= k = O(\tilde{n}^{2/3})  \\
    \tilde{d}_X &= d^2 = O(\tilde{n}^{2/3})  \\
    \tilde{d}_Z &= d = O(\tilde{n}^{1/3})  \, ,
\end{align}
\end{subequations}
where we used the fact that $d=O(\sqrt{n})$ for HGP codes.

\begin{figure}[t]
    \centering
    \includegraphics[width=0.48\textwidth]{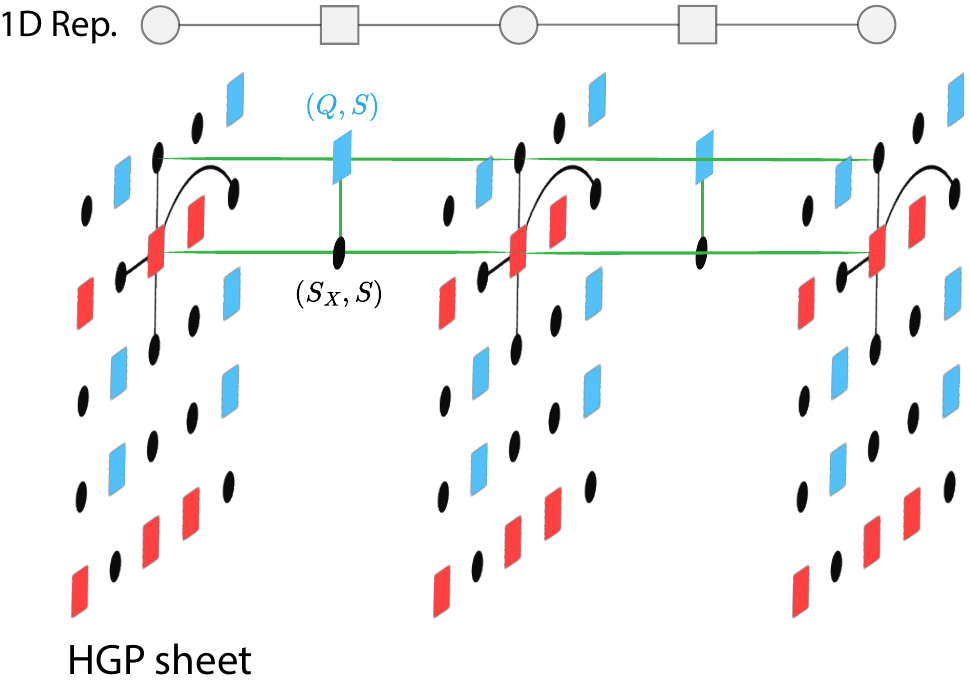}
    \caption{The thickened code \eqref{eq:3D HGP H_X and H_Z} for the HGP code in Fig. \ref{fig:2D HGP layout} with $z=3$ is illustrated. Three copies of the original HGP code are connected (green) by intermediate qubits and $Z$-checks, corresponding to the $(Q,S)$ and $(S_X,X)$ spaces of \eqref{eq:3D HGP complex}. Intermediate $Z$-checks connect pairs of identified qubits along adjacent HGP sheets, and intermediate qubits connect pairs of identified HGP $X$-checks. The 1D repetition code is overlayed above for reference. To reduce clutter, we only display a subset of intermediate qubits and $Z$-checks.}
    \label{fig:3D HGP layout}
\end{figure}

Importantly, we can utilize the metachecks \eqref{eq:3D HGP M_Z} to catch and repair any measurement errors during the transversal initialization. The end result is the logical $\ket{\overline{\bm{+}}}$ state of the thickened code up to a small residual $X$ error. We provide a bound on the residual error in Appendix \ref{app:stage 1 FT}, but we comment that it is essentially the same proof used to demonstrate single-shot error correction in sound codes \cite{Campbell_2019}; see also \cite{xu2024fastparallel}. The correction for $Z$ errors is performed in the next stage.

\subsection{Stage 2: Dimensional collapse}\label{sec:stage 2}

After the completion of stage 1, we can assume, up to a small residual error, that we are in the codespace of the thickened code. In the second stage, we begin by measuring all qubits in the $X$ basis except for one of the ``boundary'' HGP codes, which corresponds to a boundary bit of the 1D repetition code; see Fig. \ref{fig:stage 2 collapse} for an illustration. As a result, we effectively collapse the extra ``dimension'' introduced by the homological product, and we will be left with our original HGP code up to a (potentially large) $Z$ error.

To see how the dimensional collapse performs the desired code-switching, first observe that the HGP $Z$-checks on the left boundary do not extend past the boundary, and so these checks remain invariant during the bulk measurement. On the other hand, the boundary $X$-checks take the form
\begin{align}\label{eq:boundary X-checks}
    \tilde{H}_{X,\partial} = \left(\, H_X \;\big|\; \ident \,\right)  \, ,
\end{align}
where the left block denotes qubits on the boundary, and the right block denotes qubits in the adjacent intermediate sheet (green lines in Fig. \ref{fig:stage 2 collapse}). There is thus a one-to-one correspondence between each intermediate qubit and each $X$-check on the boundary. The random measurement outcomes of these intermediate qubits will determine the eigenvalues of the boundary $X$-checks. It follows that the boundary qubits are in the simultaneous eigenspace of all HGP $Z$-checks and $X$-checks, up to a potentially large $Z$ error due to the random $X$-check eigenvalues. It remains to show that the logical operators reduce to that of the HGP code on the boundary. Denote $G_X$ and $G_Z$ as the $X$-type and $Z$-type generator matrices for the HGP code, and denote $g$ as the generator matrix of the repetition code. Using the K\"unneth formula for product complexes, one can construct $X$-type and $Z$-type generator matrices for the logical operators of the thickened code \cite{bravyi2013hom}:
\begin{subequations}
\begin{align}
    \tilde{G}_Z &= \left(\, G_Z \otimes \ident_{d,1} \;\big|\; \mathbf{0} \,\right) \label{eq:tilde G_Z} \\
    \tilde{G}_X &= \left(\, G_X \otimes g \;\big|\; \mathbf{0} \,\right)  \label{eq:tilde G_X}  \, ,
\end{align}
\end{subequations}
where $\ident_{d,1}$ is a (truncated) $d \times d$ identity matrix with only the first column being nonzero. One can quickly verify that $\tilde{H}^{}_X \tilde{G}^\transpose_Z = \tilde{H}^{}_Z \tilde{G}^\transpose_X = 0$ as required. \eqref{eq:tilde G_Z} tells us that the logical $\bar{Z}$ operators of the thickened code are exactly the same as those of the HGP code, and since they live on the boundary, they survive the dimensional collapse. \eqref{eq:tilde G_X} tells us that the logical $\bar{X}$ operators of the thickened code are extended in the ``third dimension'', but their intersections with the boundary are exactly the HGP logical $\bar{X}$ operators. Thus, up to a physical $Z$ error, the dimensional collapse successfully code-switches between our thickened code and our HGP code.

At this point, in terms of fault tolerance, it seems like we have just returned to the same scenario in which we were hoping to avoid, this time with the $X$-checks. The $X$-checks of our HGP code post-collapse have random $\pm 1$ eigenvalues, and it is not immediately clear how to return to the codespace without inflicting a logical $\bar{Z}$ error. Fortunately, there is a way to correct these errors by utilizing the measurement outcomes from the bulk qubits, as we describe below.

\begin{figure}[t]
    \centering
    \includegraphics[width=0.48\textwidth]{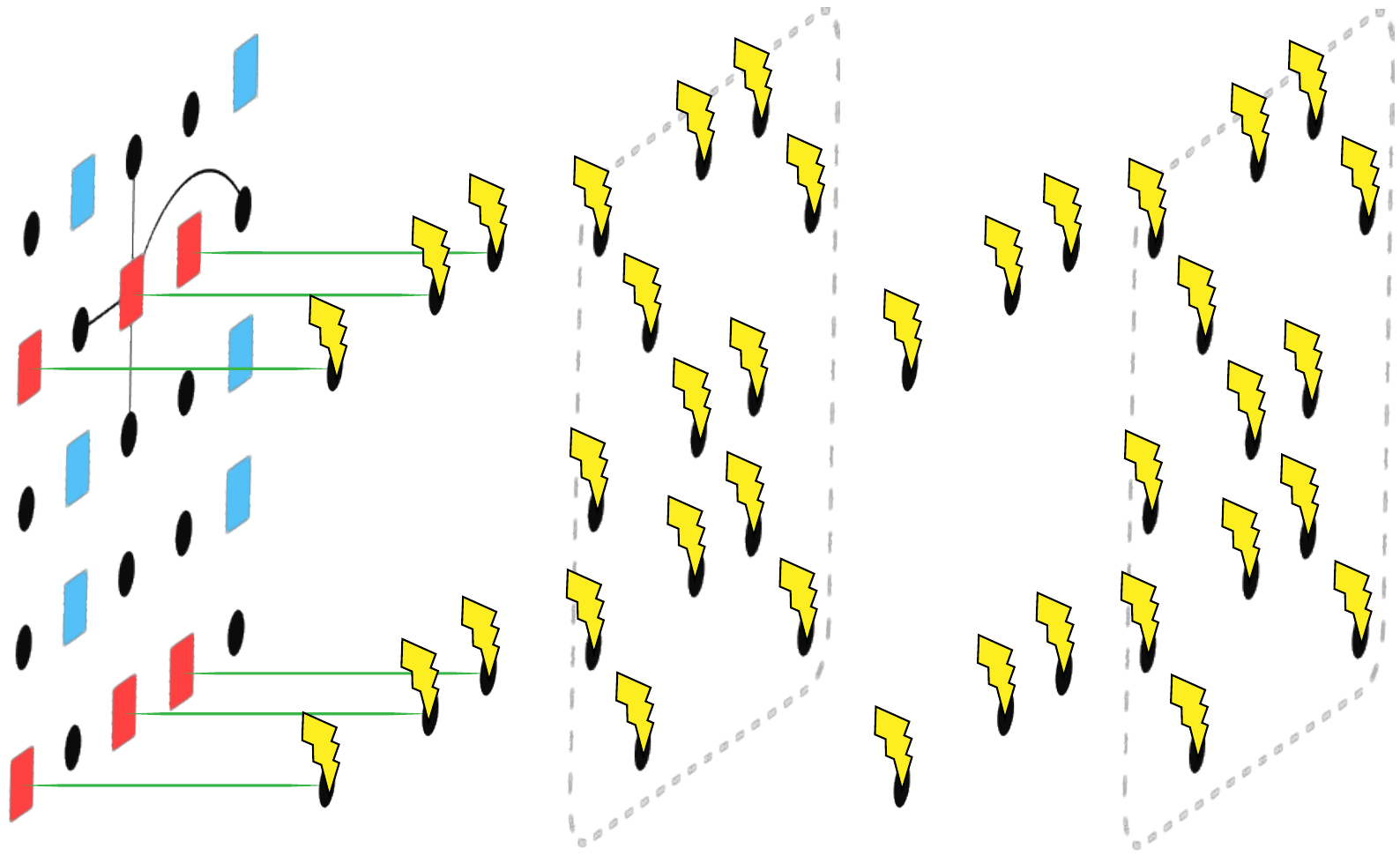}
    \caption{The dimensional collapse in stage 2 is depicted. All qubits other than a boundary HGP sheet are measured in the $X$ basis (denoted by lightning bolts). The intermediate qubits adjacent to the boundary sheet are connected by green edges to the boundary $X$-checks.}
    \label{fig:stage 2 collapse}
\end{figure}

In the state-vector picture, the logical $\ket{\overline{\bm{+}}}$ state of the thickened code has the form
\begin{align}\label{eq:logical + state}
    \ket{\overline{\bm{+}}}_{\rm th} &\propto \Bigg(\prod_{\{ \tilde{C}_Z \}} \frac{\ident+\tilde{C}_{Z}}{2}\Bigg) \ket{+}^{\otimes \tilde{n}}  \notag \\
    &\propto \sum_{\tilde{O}_Z\in\tilde{\mathcal{S}}_Z} \tilde{O}_Z \ket{+}^{\otimes \tilde{n}}  \, ,
\end{align}
where $\{\tilde{C}_{Z}\}$ is the collection of all $Z$-checks, and $\tilde{\mathcal{S}}_Z$ is the $Z$-type stabilizer subgroup of the thickened code. So we see that in the $X$-basis, the logical $\ket{\overline{\bm{+}}}$ state is a uniform superposition over all $Z$-stabilizer elements applied to $\ket{+}^{\otimes \tilde{n}}$. If we measure $X$ on all $\tilde{n}$ qubits, then we collapse the superposition onto a random basis state, which is equivalent to applying a random $Z$-stabilizer element to $\ket{+}^{\otimes \tilde{n}}$. Since we only measure a subset of all the qubits, the superposition \eqref{eq:logical + state} instead only partially collapses. Recall that the $Z$-stabilizer subgroup generated by $Z$-checks supported only on the boundary is precisely $\mathcal{S}_Z$ of our HGP code. Our post-collapse state then takes the form
\begin{align}\label{eq:post-bulk-measurement state}
    \mathcal{M}_{X,\text{bulk}} \ket{\overline{\bm{+}}}_{\rm th} &\propto \tilde{O}'_Z \sum_{O_Z \in \mathcal{S}_Z} O_Z \ket{+}^{\otimes n}_\partial \otimes \ket{+}^{\otimes (\tilde{n}-n)}_{\partial^c}  \notag \\
    &\propto \tilde{O}'_Z \ket{\overline{\bm{+}}}_\partial \otimes \ket{+}^{\otimes (\tilde{n}-n)}_{\partial^c}  \, ,
\end{align}
where $\tilde{O}'_Z \in \tilde{\mathcal{S}}_Z \setminus \mathcal{S}_Z$. In particular, the only component of $\tilde{O}'_Z$ which acts nontrivially on $\ket{\overline{\bm{+}}}_\partial$ consists of intermediate $Z$-checks with (partial) support on the boundary. Since $\tilde{O}'_Z$ arises from the randomness of measurement, we call it an \emph{intrinsic} $Z$ error.

To return to the codespace of our boundary HGP code, we need to undo the boundary action of the intrinsic $Z$ error $\tilde{O}'_Z$. We will leverage a causal structure induced by the 1D repetition code to perform this correction, generalizing a similar procedure for the 3D surface code \cite{Scruby_2022}. Starting from the opposite boundary, we locate the qubits in the $\ket{-}$ state and apply $Z$ to its identified qubit in the next sheet, which effectively pushes these $Z$ ``errors'' onto the next HGP sheet. We then iterate this procedure starting from this next sheet until we terminate at our boundary HGP code. To see why this algorithm undoes $\tilde{O}'_Z$ up to an element of $\mathcal{S}_{Z}$, we will analyze the effect of this procedure on the two types of $Z$-checks: intermediate checks and sheet checks. Since the algorithm is linear, our analysis will then also apply to arbitrary elements of $\tilde{\mathbf{S}}_Z$. Recall that an intermediate $Z$-check connects a pair of identified qubits supported on adjacent HGP sheets. So when we push our $Z$ errors between two neighboring sheets, we will cancel out its effect on the succeeding sheet. For a sheet $Z$-check located strictly on a single HGP sheet, we will push it all the way to the boundary, but since it is in the HGP stabilizer group, its effect will be trivial. Therefore, by linearity, the only nontrivial $Z$ errors that we push onto the boundary precisely correspond to intermediate $Z$-checks that touch the boundary. One can interpret the above procedure as computing the feedback corrections in a teleportation scheme \cite{Bolt_2016, Bolt_2018}.

The aforementioned algorithm only works in the absence of external errors. If we have external $Z$ errors such that the bulk pattern of $\ket{-}$ does not correspond to any $Z$-stabilizer element, we need to first correct these errors before we can proceed. Using the bulk measurement outcomes, we can infer the values of the bulk $X$-checks and reconstruct an $X$-syndrome in an attempt to correct for external $Z$ errors. However, since we did not measure our boundary sheet, what we have is an incomplete $X$-syndrome, and so we need to be careful with how to proceed with decoding. Fortunately, we can leverage the fact that our HGP code support single-shot decoding to reliably correct these $Z$ errors even with an incomplete $X$-syndrome. The crucial observation is that the bulk $X$-check Tanner graph is precisely the spacetime decoding graph for $d-1$ noisy rounds of $X$-syndrome measurement of the HGP code; the intermediate qubits label the location of syndrome measurement errors in the corresponding spacetime code \cite{Bacon_2015, gottesman2022, delfosse2023}; see Fig. \ref{fig:spacetime code} for an illustration. Using this spacetime mapping, we can reinterpret our bulk $X$-syndrome as $d-1$ rounds of noisy HGP $X$-syndromes, with the asymmetry of the bulk Tanner graph (missing boundary sheet on one side) providing us with an arrow of time. In this manner, we can employ a single-shot decoder for our HGP code to decode the bulk $X$-syndrome one sheet at a time, starting from the opposite HGP boundary and ending on our desired boundary. Furthermore, we can bake in this decoder during each step of the algorithm. The combined algorithm is described below in pseudocode.

\begin{figure}[t]
    \centering
    \includegraphics[width=0.39\textwidth]{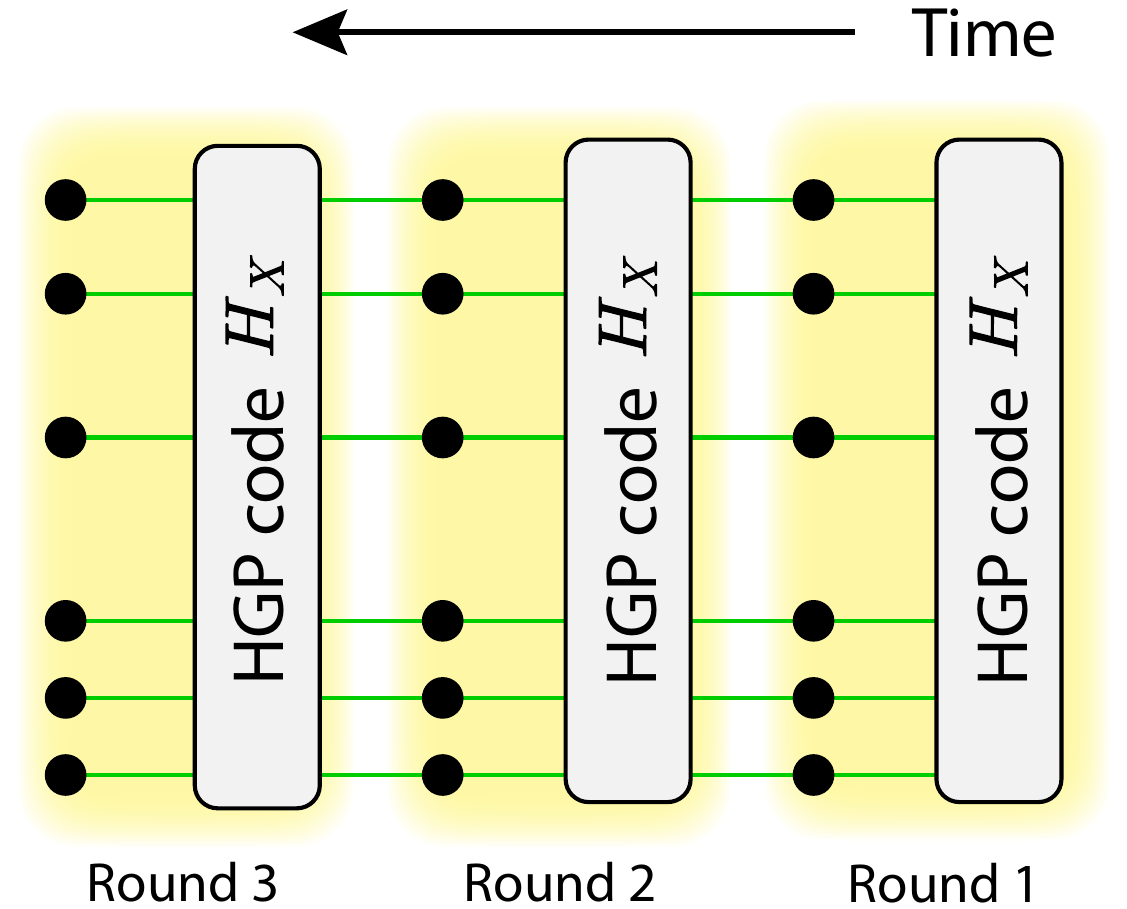}
    \caption{The interpretation of the bulk $X$-check Tanner graph as a spacetime code is depicted. Gray rectangles denote HGP sheets, and black circles denote the intermediate qubits connected to each $X$-check in its neighboring HGP sheets. The arrow of time runs from the opposite HGP boundary to the unmeasured one. The combination of each HGP sheet with its subsequent intermediate qubits can be interpreted as a round of noisy syndrome measurement.}
    \label{fig:spacetime code}
\end{figure}

\begin{algorithm}
\DontPrintSemicolon
\caption{$Z$-correction for stage 2}\label{alg:collapse correction}
\Input{bulk $X$-measurement outcomes $\{\mathbf{m}_i\}$ arranged by sheet, bulk $X$-syndromes $\{\mathbf{s}_i\}$ arranged by sheet, and a single-shot HGP decoder}
\Return{boundary $Z$-correction vector $\mathbf{z} \in \mathbb{F}^n_2$}
\;
$\mathbf{z} = \rm{zeros}(n)$  \tcp*{stored $Z$ errors}
$\hat{\mathbf{s}} = \rm{zeros}(m)$  \tcp*{inferred syndrome error}
\For{$\tau=1$ \KwTo $d-1$}{
    $\mathbf{s}' = \mathbf{s}_\tau + \hat{\mathbf{s}}$  \tcp*{$\tau=1$ is the opposite boundary}
    Decode $\mathbf{s}'$ to obtain an inferred syndrome error $\hat{\mathbf{s}}$ and an inferred $Z$ error $\hat{\mathbf{z}}$\;
    $\mathbf{z} \pluseq \mathbf{m}_\tau + \hat{\mathbf{z}}$\;
}
\end{algorithm}

As long as residual errors remain small during each round of single-shot decoding, then the final residual error that is pushed onto the boundary should also be small. A quantitative bound on the residual error is presented in Appendix \ref{app:stage 2 FT}. We note that one can also simply decode the entire bulk syndrome at once, which should further decrease any residual $Z$ error on the boundary at the cost of increased decoding complexity.

The fault tolerance of the entire protocol is analyzed in Appendix \ref{app:full protocol FT}, combining the results of Sections \ref{sec:stage 1} and \ref{sec:stage 2}.


\section{Beyond repetition}\label{sec:beyond rep}

In this section, we present a simple modification to the protocol which reduces the spatial overhead by a factor of at most 2, but at the cost of needing to initialize multiple HGP codes in parallel. Recall that for stage 1 to be fault-tolerant against syndrome measurement errors, we needed our classical code in the homological product to have the same distance as our quantum HGP code. We employed a $[d,1,d]$ repetition code and saw how its induced causality structure allowed us to formulate a decoding algorithm (Alg. \ref{alg:collapse correction}) for stage 2. The spatial overhead of the protocol essentially comes from the ratio $n/k$, which is equal to $d$ for the repetition code.

From the perspective of stage 2, we did not need the thickness to be $\ell=d$ because our HGP code was single-shot; $\ell=O(1)$ is sufficient for stage 2 to be fault-tolerant. It may be tempting to try and isolate more HGP sheets during stage 2, either by using a different classical code or by not measuring some HGP sheets, in order to lower the spatial overhead \emph{per HGP sheet}. However, naively doing so destroys the causality structure that we crucially leveraged to correct the intrinsic $Z$ error \eqref{eq:post-bulk-measurement state} on the boundary during stage 2. In particular, we may encounter a ``causal split'' in the classical Tanner graph, where there is more than one direction (check) that we can choose; see Fig. \ref{fig:causal split}. When this happens, we will encounter a $\ket{-}$ on one of our bulk HGP sheets (corresponding to a classical bit in the homological product) that could have been caused by multiple intermediate $Z$-checks (corresponding to a classical check), and without additional information it is unclear which $Z$-check we should pick.

\begin{figure}[t]
    \centering
    \includegraphics[width=0.3\textwidth]{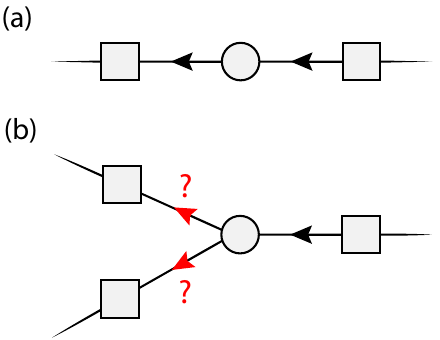}
    \caption{A causal structure is overlayed on top of two classical Tanner graphs. (a) In the repetition code, there is a clear arrow of time from right to left, which is used to correct the intrinsic $Z$ error on the left boundary during stage 2. (b) On a general Tanner graph, we may encounter a causal split, where there is no definitive path to take.}
    \label{fig:causal split}
\end{figure}

The modification we propose is to replace the 1D repetition code with a $[n,n-1,2]$ single-parity code, whose Tanner graph resembles a star. To increase the code distance, we can repeat the physical bits (i.e. concatenate with a repetition code) so that the final ``star'' code has parameters $[n,z-1,2n/z]$, where $z$ is the degree of the central check, and $n$ is a multiple of $z$. Geometrically, codewords have support on paths which begin and end on distinct branches of the star. Importantly, each physical bit is connected to at most two checks, and so we preserve the same causality structure that emerged in the repetition code. Specifically, we choose any branch as our starting point and call it the \emph{incoming} branch. We then label all other branches as \emph{outgoing} branches; see Fig. \ref{fig:star graphs}a for an illustration. When we perform the dimensional collapse in stage 2, we will now refrain from measuring the HGP code sheets located at the endpoints of all $z-1$ outgoing branches.

Now that the arrows of time have been established, we can proceed with stage 2 but with the following modifications to Algorithm \ref{alg:collapse correction}. The rule for pushing $Z$-corrections between sheets now follows the Tanner graph of the star code. In particular, when we arrive at the center of the star along the incoming branch, we will need to push our $Z$-correction onto each of the $z-1$ outgoing branches. Afterwards, the rest of the algorithm is the same as before within each outgoing branch. The correctness of this modified algorithm follows the same linearity argument given in Sec. \ref{sec:stage 2}.

\begin{figure}[t]
    \centering
    \includegraphics[width=0.31\textwidth]{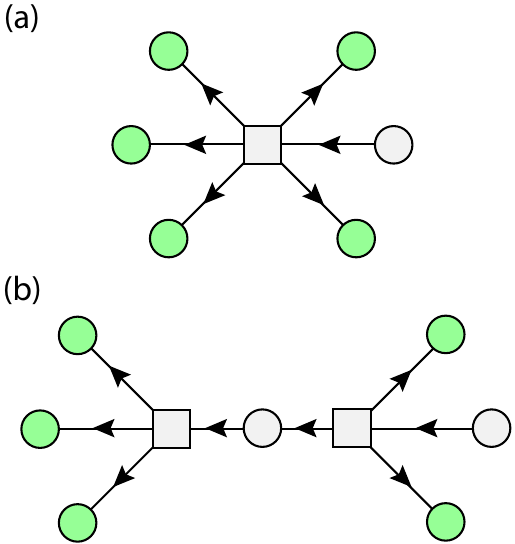}
    \caption{(a) The Tanner graph of the star code is depicted (repeated bits not shown). Green circles denote the bits which will host the final unmeasured HGP codes, and black arrows label the causal structure. (b) The same illustration is shown for a weight-reduced star code with one auxiliary bit in the center.}
    \label{fig:star graphs}
\end{figure}

\begin{table}[t]
\renewcommand{\arraystretch}{1.35}\normalsize
\begin{tabular}{|c|c|c|c|c|}
\hline
\;Code\; & $n$    & $k$   & $\frac{n}{kd}$                    & Weight \\ \hline
1D Rep.  & $d$    & $1$   & $1$                               & $2$        \\
Star           & \;$zd/2$\; & \;$z-1$\; & \;$\frac{1}{2} \cdot \frac{z}{z-1}$\; & $z$        \\ \hline
\end{tabular}
\caption{Parameters of the 1D repetition code and the degree-$z$ star code are tabulated. $k$ is also the number of HGP codes produced at the end of the protocol. The penultimate column displays the overhead reduction compared to the repetition code, and the last column displays the maximum check weight.}
\label{tab:rep vs star}
\end{table}

We can now compute the overhead reduction of the star code over the repetition code with fixed distance $d$. We use a heuristic $n/(kd)$, which equals one for the repetition code and is less than one for more efficient encodings. For our $[n,z-1,2n/z]$ star code, we have
\begin{align}
    \frac{n}{kd} = \frac{1}{2} \cdot \frac{z}{z-1} > \frac{1}{2} \, .
\end{align}
In particular, the overhead reduction per HGP code approaches $1/2$ for increasing $z$. A comparison between the star code and the repetition code is catalogued in Table \ref{tab:rep vs star}. Although stage 1 fault tolerance may begin to suffer at large $z$ due to deeper syndrome extraction circuits, one can perform weight-reduction \cite{LRESC, sabo2024weight} on the star code to evenly distribute the central check weight amongst multiple checks, at the cost of extra auxiliary bits; see Fig. \ref{fig:star graphs}b for an example. Since the weight-reduced code resembles a ``stretched'' version of the original code, with auxiliary bits connected in a chain-like fashion, it inherits the causal structure as well.

Lastly, we need to ensure that each HGP code is indeed in its logical $\ket{\overline{\bm{+}}}$ state at the end of the protocol. We can show this result by examining the codewords of the star code. We can write the generator matrix of the $[n,n-1,2]$ code in the form
\begin{align}\label{eq:g_star}
    g_\star = \left(\, \ident_z \;\big|\; \mathbf{1}_{z\times 1} \,\right) \, ,
\end{align}
where $\ident_z$ is the $z\times z$ identity matrix, and $\mathbf{1}_{z\times 1}$ is a column of all ones corresponding to the starting branch. For the thickened code, we can write its $X$-type generator matrix as
\begin{align}\label{eq:tilde G_X star}
    \tilde{G}_X = \left(\, G_X \otimes g_\star \;\big|\; \mathbf{0} \,\right) \, ,
\end{align}
where $G_X$ is the $X$-type generator matrix for the HGP code, analogous to \eqref{eq:tilde G_X}. From \eqref{eq:g_star} and \eqref{eq:tilde G_X star}, we see that for each HGP logical $\bar{X}$ operator, there exist $z-1$ distinct logical $\bar{X}$ operators in the thickened code supported on HGP sheets labeled by the incoming and outgoing branches of the star code. Since each branch hosts a distinct logical operator, and our initial logical state of the thickened code was the product $\ket{\overline{\bm{+}}}$ state, the resulting state on the $z-1$ HGP codes is $\ket{\overline{\bm{+}}}^{\otimes (z-1)}$.

Although the overhead reduction with the star code is only a constant, it applies at all distances (not just an asymptotic result). So if one desires to prepare multiple HGP code blocks in parallel, the star code provides an immediate spatial reduction, which may be beneficial for moderate system sizes. From Fig. \ref{fig:star graphs}, we can see that the causality structure and arrows of time can be interpreted as a type of message-passing schedule on the Tanner graph. For tree-like Tanner graphs such as that of the star code, it is quite easy to establish a schedule as we have seen. However, the code parameters of Tanner graphs without loops is extremely limited \cite{Etzion_1999}. For more general Tanner graphs, one may need to worry about scheduling around closed loops and choosing endpoints carefully. In particular, an endpoint vertex cannot have any outgoing arrows since we do not measure its corresponding HGP sheet.


\section{Numerical simulations}\label{sec:numerics}

\begin{figure*}[t]
    \centering
    \includegraphics[width=0.45\textwidth]{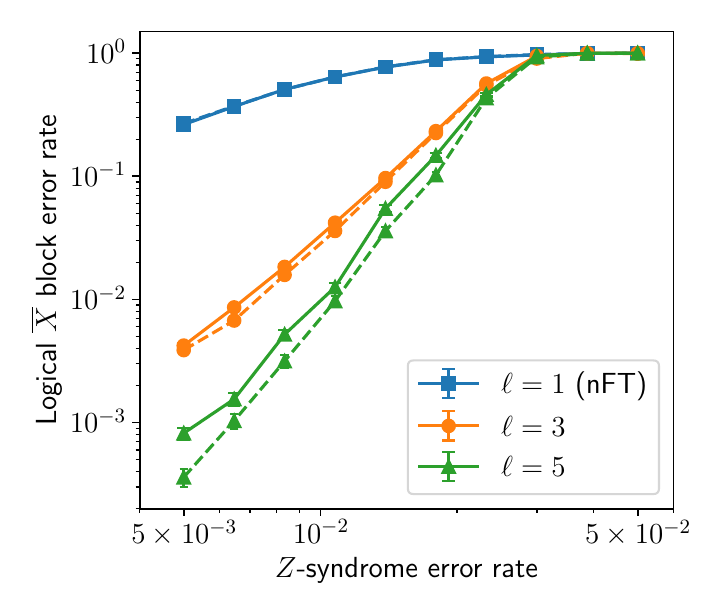}
    \hspace{2em}
    \includegraphics[width=0.45\textwidth]{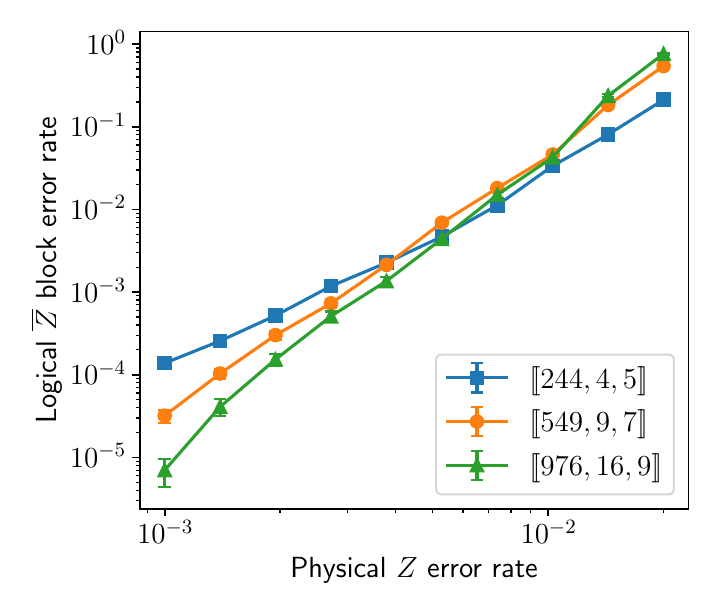}
    \caption{\textbf{Left:} The logical $\bar{X}$ error rate as a function of the syndrome error rate is plotted for the first simulation of a $\llbracket 549,9,9 \rrbracket$ HGP code with varying thickness $\ell$ (solid lines). The same logical error rate for $\ell$ rounds of repeated syndrome measurements is also displayed for comparison (dashed lines). \textbf{Right:} The logical $\bar{Z}$ error rate as a function of the physical $Z$ error rate is plotted for the second simulation of three HGP codes derived from the $(5,6)$-LDPC ensemble. Error bars denote the standard error $\sigma = \sqrt{p(1-p)/N}$ for $N$ samples.}
    \label{fig:numerics}
\end{figure*}

In the fault-tolerance analysis of the previous section, we had assumed both an adversarial error model as well as theoretically optimal decoders. In practice, we will mostly likely encounter a local stochastic error model, such as independent single-qubit errors. We will also need an efficient decoder since it is known that maximum-likelihood decoding of linear codes is generically NP hard \cite{Berlekamp_1978}.

We numerically benchmark the performance of the protocol under a phenomenological noise model: each physical qubit has an independent probability $p$ to incur a Pauli $X$ or $Z$ error, and each $Z$-check in the first stage is incorrectly measured with probability $p$. For syndrome decoding, we employ an open-source message-passing decoder \cite{Roffe_LDPC_Python_tools_2022} based on belief propagation with ordered-statistics postprocessing (BP+OSD) \cite{Panteleev_BPOSD_2021}, which displays good performance as a general-purpose quantum LDPC decoder \cite{Panteleev_BPOSD_2021, Roffe_decoding_2020, Higgott_2023}. The time complexity of BP is $O(n)$ while that of OSD is $O(n^3)$. We configure BP+OSD with the ``minimum-sum'' strategy (20 iterations) for BP and the ``combination-sweep'' strategy (search depth 20) for OS.

For the first simulation, we construct a $\llbracket 549,9,9 \rrbracket$ HGP code by taking the hypergraph product of a classical $[18,3,9]$, $(5,6)$-LDPC code with itself; the classical parity-check matrix is randomly selected using the bipartite configuration model. We then construct the thickened code \eqref{eq:3D HGP H_X and H_Z} by taking a homological product of this HGP code with a $[\ell,1,\ell]$ repetition code, where $\ell$ is a tunable thickness parameter. Without loss of generality, we assume that the true initial $Z$-syndrome is the all-zero vector. We then flip the value of each $Z$-check in the syndrome with probability $p$ to obtain a corrupted syndrome. We repair this corrupted syndrome using BP+OSD with respect to the $Z$-metacheck matrix \eqref{eq:3D HGP M_Z}, and then decode the repaired syndrome using BP+OSD with respect to $\tilde{H}_Z$ \eqref{eq:3D HGP H_Z} to obtain an $X$-correction. Since our true initial syndrome was zero by assumption, this $X$-correction hence becomes the residual $X$ error. To account for the second stage of the protocol, we take the boundary restriction of this residual $X$ error so that what remains is our original HGP code with a residual $X$ error. To verify whether this residual error is dangerous, we apply a new round of independent $X$ errors (with probability $p$ per qubit), add it to the residual error, and finally decode this combined error using BP+OSD with respect to $H_Z$ \eqref{eq:HGP H_Z}. We declare success if no logical $\bar{X}$ is applied at the end and failure otherwise. The simulation results are showcased on the left plot of Fig. \ref{fig:numerics}. Note that $\ell=1$ corresponds to non-fault-tolerant (nFT) transversal initialization of the original HGP code with no syndrome repair. Below $p \lesssim 3\times 10^{-2}$, we observe significant suppression of logical errors upon increasing the thickness $\ell$, providing evidence that the metachecks are successfully controlling residual errors. We also compare our protocol with the traditional method of repeated syndrome measurements and observe nearly identical performance.

For the second simulation, we randomly select three codes of increasing length from the $(5,6)$-LDPC ensemble and take the hypergraph product of each code with itself; each resulting HGP code has rate $k/n = 1/61 \approx 1.64\%$. For each HGP code, we construct its corresponding thickened code by taking the homological product of the HGP code with a classical repetition code whose length is equal to the distance of the HGP code. To simulate the decoding process for each thickened HGP code, we apply a $Z$ error to each physical qubit with probability $p$, compute the bulk error syndrome, and then run the decoding strategy given in Algorithm \ref{alg:collapse correction}, which will produce a final residual error on our boundary HGP code. For the single-shot BP+OSD decoder, we use the augmented parity-check matrix \cite{Higgott_2023}
\begin{align}
    H^{ss}_X = \left(\, H_X \;\big|\; \ident \,\right)  \, ,
\end{align}
where we connect an auxiliary qubit to each $X$-check. Corrections on the auxiliary qubits denote inferred syndrome errors. Note the similarity to \eqref{eq:boundary X-checks}. After Algorithm \ref{alg:collapse correction} has terminated, we probe the residual error with a final round of perfect syndrome decoding on the HGP code and verify whether a logical error has occurred. These simulation results are showcased on the right plot of Fig. \ref{fig:numerics}. For the family of HGP codes built from the $(5,6)$-LDPC ensemble, we observe threshold-like behavior around a physical $Z$ error rate of $p \approx 6 \times 10^{-3}$. We did not observe any definitive evidence of a threshold when using input LDPC ensembles with degrees less than $(5,6)$ and similar system sizes.

To conclude this section, we briefly comment on the classical decoding complexity of the entire protocol under BP+OSD. Recall that the thickened code has $\tilde{n} = \mathrm{\Theta}(n^{3/2})$ \eqref{eq:tilde n} and $\dim \tilde{M}_Z = \mathrm{\Theta}(\tilde{n})$ \eqref{eq:3D HGP M_Z}. Hence the syndrome repair step during stage 1 has time complexity $O(n^{9/2})$ from OSD. If we run the sequential decoder of Algorithm \ref{alg:collapse correction} for the second stage, then the decoding complexity is $O(n^3\sqrt{n}) = O(n^{7/2})$. If we decoded the entire bulk syndrome at once, then the decoding complexity is $O(n^{9/2})$. Parallelized algorithms leveraging LDPC properties \cite{wolanski2024, BPLSD, BPOTF} may significantly reduce the decoding complexity for low error rates.


\section{Outlook}\label{sec:outlook}

We have designed a single-shot $\ket{\overline{\mathbf{0}}} / \ket{\overline{\bm{+}}}$ state-preparation protocol for constant-rate HGP codes and demonstrated its fault tolerance against adversarial noise. The construction is inspired by the notion of dimension-jumping in 2D topological codes \cite{Bombin_2016, Scruby_2022}, and we provide a direct generalization to HGP codes using the homological product. The key ingredient is a mapping between a thickened code and a spacetime code, induced by the homological product with a structured classical code. This mapping allows us to import existing single-shot HGP decoders \cite{Fawzi_2018, Higgott_2023} to perform fault-tolerant code-switching during the second stage of the protocol.

Since the HGP codes we study are LDPC, the circuit depth of our protocol (with enough ancillas) is $O(1)$, which may provide significant improvements to logical clock speeds in fault-tolerant architectures where state preparation is a key subroutine \cite{xu2024fastparallel}. Unfortunately, the protocol does not preserve the constant-rate property of HGP codes, which is often desirable from a physical resource standpoint \cite{gottesman2014, Xu_2024_constant}. In terms of a combined spacetime overhead, our protocol essentially trades the original $O(d)$ time overhead with an $O(d)$ spatial overhead, and so we only expect the protocol to be practical in the regime where computational time becomes more valuable than physical qubit number. While not a completely satisfactory solution, the protocol is a first step in the direction of efficiently preparing HGP codewords, and we hope that the ideas set forth in this paper accelerate progress along this research thrust. In particular, it would be interesting to explore generalizations of Sec. \ref{sec:beyond rep}, such as whether other graphs could further lower the spatial overhead, perhaps even to a constant.

We also emphasize that our protocol is not limited to just HGP codes. Although the thickening procedure \eqref{eq:3D HGP complex} applies to any CSS code, we require a special structure in its 3-term chain \eqref{eq:CSS 3-term chain} for fault tolerance. For stage 1 to be fault-tolerant against syndrome measurement errors, we required a parity-check matrix of the thickened code to exhibit soundness. Since HGP codes are formed from the homological product, their classical ``transpose'' codes (by relabeling the CSS 3-term chain \eqref{eq:CSS 3-term chain} as \eqref{eq:LTC 3-term chain}, which can also be interpreted as a gauging procedure \cite{rakovszky2023ldpc}) are sound. This soundness is then inherited by the thickened code upon taking an additional homological product \cite{Campbell_2019}. So in order for us to use this protocol for some generic CSS code, we need its classical transpose code to be sound. Families of good quantum LDPC codes satisfy this property since their classical transpose codes are locally testable with linear soundness functions \cite{Panteleev_2022_LDPC, qTanner_codes, Dinur_2023_LDPC, Dinur_2022_LTC}. The second important ingredient is a single-shot decoder for stage 2, which has been demonstrated with the quantum Tanner codes \cite{Gu_2024}. Alternatively, we could have simply just used a constant-rate, four-dimensional homological product code which provides both $O(1)$ spatial and temporal overheads \cite{xu2024fastparallel}. However, these $O(1)$ constants are rather large in practice (e.g. rate gets decreased by a power of 4), and it remains to be seen whether higher-dimensional homological product codes can rival their two-dimensional cousins in finite-length performance.

It would also be interesting to see whether a sustainable threshold for the protocol can be proven for local stochastic noise. Even though for any local error rate there are $O(n)$ errors on average, because the code is LDPC, typical error configurations form separated clusters that can be decoded independently \cite{Kovalev_2013}. Since HGP codes can have linear confinement \cite{Leverrier_2015}, they do exhibit thresholds under local stochastic noise \cite{Quintavalle_2021}. As a result, we already know that stage 2 has a threshold under local stochastic noise due to the spacetime mapping, and it is precisely the threshold for single-shot error correction. It remains to be seen whether a sustainable threshold can be proven for stage 1. Our numerical simulations in Sec. \ref{sec:numerics} suggest that a sustainable threshold should exist. It has also recently been shown that higher-dimensional hypergraph product codes (homological products of 1-complexes) support distance-preserving syndrome extraction circuits \cite{tan2024hgp}, and so we further expect competitive performance under circuit-level noise.

Lastly, we speculate that our protocol could potentially provide an avenue to non-Clifford gates in HGP codes. Since the logical operators of HGP codes essentially mimick their classical codewords along horizontal and vertical directions (in a 2D layout), they can be chosen such that logical $\bar{X}$ and $\bar{Z}$ only intersect on at most one physical qubit. As a result, transversal gates will be limited to the Clifford group \cite{Burton_2022} just like for 2D codes \cite{Bravyi_2013_transversal}. However, for the surface code, there exist constructions in 3D which enable transversal non-Clifford gates \cite{Kubica_2015, Vasmer_2019_3D}, and one can use dimension-jumping to apply these ideas to 2D \cite{Brown_2020, Scruby_2022_jit, Scruby_2022}, thereby completing a universal gate set. It would interesting to see if the thickening procedure can be modified such that the thickened code admits transversal non-Clifford gates, which combined with single-shot dimensional collapse, may provide a route to circumvent magic state distillation in constant-rate HGP codes.

\section*{Acknowledgements}

We thank Victor Albert for pointing out the connection to spacetime codes. We also thank Daniel Gottesman and Michael Gullans for valuable discussions, and especially Noah Berthusen and Shi Jie Samuel Tan for feedback on the draft. This work is sponsored by the Department of Energy Office of Advanced Scientific Computing Research's (DoE ASCR) Quantum Testbed Pathfinder program via Award No. DE-SC0019040 and No. DE-SC0024220. We also acknowledge support from the DARPA MeasQuIT program.

\section*{Data availability}

All simulation code and data are available at \href{https://github.com/yifanhong/single-shot-HGP-prep}{this GitHub repository}.


\newpage
\appendix
\renewcommand{\thesubsection}{\thesection.\arabic{subsection}}

\section{Fault-tolerance analysis of stage 1}\label{app:stage 1 FT}

We first need to introduce some extra notation. Recall the 4-term chain of \eqref{eq:3D HGP complex}:
\begin{align}
    S_X \xlongrightarrow{\tilde{H}^\transpose_X} Q \xlongrightarrow{\tilde{H}_Z} S_Z \xlongrightarrow{\tilde{M}_Z} R_Z \, .
\end{align}
Define the single-shot distance $d_{\rm ss}$ as the minimum weight of nontrivial elements in the second homology group $\mathcal{H}_2 \equiv \ker \tilde{M}_Z / \im \tilde{H}_Z$, comprised of $Z$-syndromes that satisfy all $Z$-metachecks but cannot be produced by any physical $X$ error. By convention, we set $d_{\rm ss} = \infty$ if $\mathcal{H}_2$ is trivial.

Denote $\tilde{m} = md + n(d-1)$ as the number of $Z$-checks in the thickened code (i.e. rows of $\tilde{H}_Z$). We now prove the following Lemma.

\begin{lem}[Stage 1 residual $X$ error]
\label{lem:stage 1 residual}
    For any $Z$-syndrome error $\mathbf{s}_{\rm e} \in \mathbb{F}^{\tilde{m}}_2$ during the stage 1 protocol that satisfies $\abs{\mathbf{s}_{\rm e}} \leq d/2$, the residual $X$ error  $\mathbf{e}+\hat{\mathbf{e}} \in \mathbb{F}^{\tilde{n}}_2$ satisfies
    \begin{align}
        \norm{\mathbf{e}+\hat{\mathbf{e}}} \leq f\big( 2\abs{\mathbf{s}_{\rm e}} \big) \, ,
    \end{align}
    where $f(x) = x^3/4$.
\end{lem}

\begin{proof}
    Lemma 5 and Lemma 6\footnote{Technically, Lemma 6 of \cite{Campbell_2019} involves the homological product of two 3-term chains, whereas we take the product of a 3-term chain with a 2-term chain. However, the one-sided version of the statement (and Lemma 7) still holds upon setting one of the boundary maps of one of the chains (e.g. $\tilde{\delta}_0$) as the zero map, which results in $r_c=0$ in the proof.} of \cite{Campbell_2019} tell us that $\tilde{H}_Z$ of our thickened code \eqref{eq:3D HGP H_Z} is $(d,x^3/4)$-sound with $d_{\rm ss} = \infty$ according to Def. \ref{defn:soundness}. In the first step of decoding, we find a minimum-weight syndrome correction $\hat{\mathbf{s}}$ such that the total syndrome $\mathbf{s} + \mathbf{s}_{\rm e} + \hat{\mathbf{s}} \in \ker\tilde{M}_z$, i.e. it satisfies all metachecks. Since $\mathbf{s} \in \im\tilde{H}_Z \subseteq \ker\tilde{M}_z$, by linearity we also have that $\mathbf{s}_{\rm e} + \hat{\mathbf{s}} \in \ker\tilde{M}_z$. Furthermore, because $\mathcal{H}_2 = \ker \tilde{M}_Z / \im \tilde{H}_Z$ is trivial ($d_{\rm ss} = \infty$), $\mathbf{s}_{\rm e} + \hat{\mathbf{s}} \in \im \tilde{H}_Z$, i.e. it is a physical syndrome.

    Having established that $\mathbf{s} + \mathbf{s}_{\rm e} + \hat{\mathbf{s}}$ is a physical $Z$-syndrome, we now find a minimum-weight $X$-correction $\hat{\mathbf{e}}$ such that $\tilde{H}_Z \hat{\mathbf{e}} = \mathbf{s} + \mathbf{s}_{\rm e} + \hat{\mathbf{s}}$. Using the fact that $\mathbf{s} = \tilde{H}_Z \mathbf{e}$, we can write
    \begin{align}
        \mathbf{s}_{\rm e} + \hat{\mathbf{s}} &= \mathbf{s} + \tilde{H}_Z \hat{\mathbf{e}}  \notag \\
        &= \tilde{H}_Z \mathbf{e} + \tilde{H}_Z \hat{\mathbf{e}}  \notag \\
        &= \tilde{H}_Z \left( \mathbf{e} + \hat{\mathbf{e}} \right) \, .
    \end{align}
    We will now use the soundness property to upper bound the minimum weight of $\mathbf{e} + \hat{\mathbf{e}}$. Note that since $\hat{\mathbf{s}}$ is a minimum-weight syndrome correction, its weight is at most that of $\mathbf{s}_{\rm e}$ because $\hat{\mathbf{s}} = \mathbf{s}_{\rm e}$ is a possible solution. Thus, we have
    \begin{align}
        \abs{\mathbf{s}_{\rm e} + \hat{\mathbf{s}}} \leq \abs{\mathbf{s}_{\rm e}} + \abs{\hat{\mathbf{s}}} \leq 2\abs{\mathbf{s}_{\rm e}} \leq d  \, ,
    \end{align}
    and so $\mathbf{s}_{\rm e} + \hat{\mathbf{s}}$ lies within the soundness radius of the code. From the definition of soundness (Def. \ref{defn:soundness}), there then exists an $\mathbf{e} + \hat{\mathbf{e}}$ that satisfies
    \begin{align}
        \norm{\mathbf{e} + \hat{\mathbf{e}}} \leq f\big( \abs{\mathbf{s}_{\rm e} + \hat{\mathbf{s}}} \big) \leq f\big( 2\abs{\mathbf{s}_{\rm e}} \big)  \, ,
    \end{align}
    where in the last inequality we used the fact that $f(x)=x^3/4$ is convex.
\end{proof}


\section{Fault-tolerance analysis of stage 2}\label{app:stage 2 FT}

The idea will be to use the confinement property of constant-rate HGP codes \cite{Leverrier_2015} to show that the residual error in each iteration of Algorithm \ref{alg:collapse correction} is bounded. The proof utilizes the single-shot capability of confined codes, which has been shown in \cite{Quintavalle_2021}; we recite the relevant results for completeness.

We employ the Shadow decoder \cite{Quintavalle_2021} which takes in a parameter $t>0$ and operates in two steps: given a corrupted syndrome $\mathbf{s} + \mathbf{s}_{\rm e}$, 
\begin{enumerate}
    \item Find a minimum-weight syndrome correction $\hat{\mathbf{s}}$ such that for the repaired syndrome $\mathbf{s} + \mathbf{s}_{\rm e} + \hat{\mathbf{s}}$, there exists an error $\mathbf{e}'$ with $\norm{\mathbf{e}'} \leq t$ satisfying $H\mathbf{e}' = \mathbf{s} + \mathbf{s}_{\rm e} + \hat{\mathbf{s}}$.
    \item Find a minimum-weight correction $\hat{\mathbf{e}}$ such that $H\hat{\mathbf{e}} = \mathbf{s} + \mathbf{s}_{\rm e} + \hat{\mathbf{s}}$.
\end{enumerate}

If $H$ exhibits confinement, then the weight of the residual error under the Shadow decoder is bounded, as formalized in the following lemma.

\begin{lem}[Lemma 1 of \cite{Quintavalle_2021}]
\label{lem:confined residual}
    Consider a parity-check matrix $H$ that is $(t,f)$-confined. For any error $\mathbf{e}$ with $\norm{\mathbf{e}} \leq t/2$ and syndrome error $\mathbf{s}_{\rm e}$, the residual error $\mathbf{e}+\hat{\mathbf{e}}$ left by the shadow decoder with parameter $t/2$ is bounded by
    \begin{align}\label{eq:confined residual}
        \norm{\mathbf{e}+\hat{\mathbf{e}}} \leq f\big( 2\abs{\mathbf{s}_{\rm e}} \big) \, .
    \end{align}
\end{lem}
The proof of Lemma \ref{lem:confined residual} is of the same flavor as the one for sound codes. Recall that for stage 2, we reinterpret the bulk $X$-check Tanner graph as a spacetime decoding graph over $d-1$ QEC rounds of the HGP code. In particular, $Z$ errors on sheet qubits denote physical $Z$ errors, and $Z$ errors on intermediate qubits denote syndrome errors. To make this mapping explicit, for a bulk error $\mathbf{e} \in \mathbb{F}^{(\tilde{n}-n)}_2$, denote the partitions
\begin{align}\label{eq:bulk error partition}
    E \equiv \mathbf{e} \vert^{}_{\rm sheets} \in \mathbb{F}^{(d-1)\times n}_2 \;\;,\;\; S_{\rm e} \equiv \mathbf{e} \vert^{}_{\rm int} \in \mathbb{F}^{(d-1)\times m}_2
\end{align}
as matrices for the physical and syndrome errors where the first index runs over the $d-1$ sheets in the bulk. We also construct the matrix $S \in \mathbb{F}^{(d-1)\times n}_2$ similarly but with the $X$-syndromes. With slight abuse of notation, we will denote $\norm{E}$ as the Hamming weight of $E$ when $E$ is the restriction of the reduced error (i.e. minimum-weight representative of $\mathbf{e} + \mathcal{S}_Z$).

A key subroutine of Algorithm \ref{alg:collapse correction} involves single-shot decoding of the original HGP code. We will use the fact that constant-rate HGP codes exhibit linear confinement, i.e. $f(x) = \alpha x$ and $t=\gamma\sqrt{n}$, when their underlying Tanner graphs have sufficient expansion \cite{Leverrier_2015}. As a result, Lemma \ref{lem:confined residual} guarantees their single-shot performance under the Shadow decoder. The next lemma says that, if both the physical and syndrome errors in the bulk are sufficiently small, then the residual error on the boundary at the end of stage 2 is bounded.

\begin{lem}[Stage 2 residual $Z$ error]
\label{lem:stage 2 residual}
    Suppose $H_X$ of the HGP code is $(t,f)$-confined. Given a $Z$ error $\mathbf{e} \in \mathbb{F}^{\tilde{n}-n}_2$ in the bulk whose partitions \eqref{eq:bulk error partition} satisfy 
    \begin{align}
        \norm{E} \leq \frac{t}{4} \quad,\quad f\big( 2\abs{S_{\rm e}} \big) \leq \frac{t}{4} \, ,
    \end{align}
    the residual $Z$ error on the boundary $\mathbf{e}_\partial \in \mathbb{F}^n_2$ at the end of the stage 2 protocol satisfies
    \begin{align}\label{eq:boundary residual small}
        \norm{\mathbf{e}_\partial} \leq t/4 \, .
    \end{align}
\end{lem}

\begin{proof}
    Following Algorithm \ref{alg:collapse correction}, we define
    \begin{align}
        \mathbf{e}^\tau \equiv E_\tau \in \mathbb{F}^n_2 \;,\;\; \mathbf{s}^\tau \equiv S_\tau \in \mathbb{F}^m_2 \;,\;\; \mathbf{s}^\tau_{\rm e} \equiv S_{\mathrm{e},\tau} \in \mathbb{F}^m_2
    \end{align}
    as the restrictions of $E$, $S$, and $S_{\rm e}$ to the sheet indexed by $\tau=1,\dots,d-1$, where $\tau=1$ denotes the opposite boundary. Since $\mathbf{e}^\tau$ and $\mathbf{s}^\tau$ are restrictions, we have $\norm{\mathbf{e}^\tau} \leq \norm{E} \leq t/4$ and $f(\abs{\mathbf{s}^\tau_{\rm e}}) \leq f(\abs{S_{\rm e}}) \leq t/4$. Denote $\bar{\mathbf{e}}^\tau$ and $\bar{\mathbf{s}}^\tau$ as the physical error and observed syndrome respectively prior to decoding in step $\tau$ of Algorithm \ref{alg:collapse correction}. The Shadow decoder takes in $\bar{\mathbf{s}}^\tau$ and outputs a syndrome correction $\hat{\mathbf{s}}^\tau \in \mathbb{F}^m_2$. In the next step of the algorithm, this syndrome correction $\hat{\mathbf{s}}^\tau$ is then added to the next syndrome $\mathbf{s}^{\tau+1}$ and its syndrome error $\mathbf{s}^{\tau+1}_{\rm e}$ to obtain the next observed syndrome $\bar{\mathbf{s}}^{\tau+1} = \hat{\mathbf{s}}^\tau + \mathbf{s}^{\tau+1} + \mathbf{s}^{\tau+1}_{\rm e}$.
    
    We now proceed to prove \eqref{eq:boundary residual small} by induction. Suppose the residual error at step $\tau$ satisfies $\norm{\bar{\mathbf{e}}^\tau + \hat{\mathbf{e}}^\tau} \leq t/4$. Then the combined physical error in the next step $\bar{\mathbf{e}}^{\tau+1}$ obeys
    \begin{align}\label{eq:residual error tau+1}
        \norm{\bar{\mathbf{e}}^{\tau+1}} &= \norm{\mathbf{e}^{\tau+1} + \bar{\mathbf{e}}^\tau + \hat{\mathbf{e}}^\tau}  \notag \\
        &\leq \norm{\mathbf{e}^{\tau+1}} + \norm{\bar{\mathbf{e}}^\tau + \hat{\mathbf{e}}^\tau}  \notag \\
        &\leq t/2 \, ,
    \end{align}
    which satisfies the assumption of Lemma \ref{lem:confined residual}. Plugging in the result of the lemma \eqref{eq:confined residual}, we arrive at
    \begin{align}
        \norm{\bar{\mathbf{e}}^{\tau+1} + \hat{\mathbf{e}}^{\tau+1}} \leq f\big( 2\abs{\mathbf{s}^{\tau+1}_{\rm e}} \big) \leq f\big( 2\abs{S_{\rm e}} \big) \leq t/4 \, .
    \end{align}
    Since the physical error in the initial step ($\tau=1$ sheet) obeys $\norm{\bar{\mathbf{e}}^1} = \norm{\mathbf{e}^1} \leq \norm{E} \leq t/4$, \eqref{eq:residual error tau+1} holds for $\tau=1$, and thus by induction we conclude that
    \begin{align}
        \norm{\mathbf{e}_\partial} = \norm{\bar{\mathbf{e}}^{d-1} + \hat{\mathbf{e}}^{d-1}} \leq t/4 \, .
    \end{align}
\end{proof}


\section{Fault-tolerance analysis of the full protocol}
\label{app:full protocol FT}

We now combine the results of Appendix \ref{app:stage 1 FT} and Appendix \ref{app:stage 2 FT} in the following theorem which demonstrates the single-shot capability of the entire protocol.

\begin{thm}[Full protocol residual error]
    Suppose we have a HGP code with $(t,f)$-confinement and distance $d>t$. Let $\mathbf{s}_{\rm e} \in \mathbb{F}^{\tilde{m}}_2$ be a $Z$-syndrome error during stage 1, and let $\tilde{\mathbf{e}} \in \mathbb{F}^{\tilde{n}-n}_2$ be a bulk $Z$ error during stage 2 with partitions \eqref{eq:bulk error partition}. If
    \begin{align}
        \abs{\mathbf{s}_{\rm e}} \leq \frac{d}{2} \;\;,\quad \norm{E} \leq \frac{t}{4} \;\;,\quad f\big( 2\abs{S_{\rm e}} \big) \leq \frac{t}{4} \, ,
    \end{align}
    then the residual $X$ and $Z$ errors $\mathbf{e}_X$ and $\mathbf{e}_Z$ on the HGP code at the end of the protocol satisfy
    \begin{align}\label{eq:full protocol residuals}
        \norm{\mathbf{e}_X} \leq 2\abs{\mathbf{s}_{\rm e}}^3 \quad,\quad \norm{\mathbf{e}_Z} \leq \frac{d}{4} \, .
    \end{align}
\end{thm}

\begin{proof}
    Since $\abs{\mathbf{s}_{\rm e}} \leq d/2$, Lemma \ref{lem:stage 1 residual} tells us that the bulk residual $X$ error $\tilde{\mathbf{e}}_X \in \mathbb{F}^{\tilde{n}-n}_2$ satisfies
    \begin{align}
        \norm{\tilde{\mathbf{e}}_X} \leq \frac{1}{4}\left(2\abs{\mathbf{s}_{\rm e}}\right)^3 = 2\abs{\mathbf{s}_{\rm e}}^3 \, .
    \end{align}
    Because $\mathbf{e}_X \subset \tilde{\mathbf{e}}_X$ is a restriction onto the boundary, we arrive at the first equation in \eqref{eq:full protocol residuals}. Since $\norm{E}, f\big( 2\abs{S_{\rm e}} \big) \leq t/4$, the second equation in \eqref{eq:full protocol residuals} follows from Lemma \ref{lem:stage 2 residual} and the fact that $d>t$ from the definition of confinement (Def. \ref{defn:confinement}).
\end{proof}

In particular, if $2\abs{\mathbf{s}_{\rm e}}^3 < d/2$, then the residual errors on the boundary are correctable (by minimum-weight decoding).

\bibliographystyle{apsrev4-2}
\bibliography{thebib}

\begin{thebibliography}{94}%
\makeatletter
\providecommand \@ifxundefined [1]{%
 \@ifx{#1\undefined}
}%
\providecommand \@ifnum [1]{%
 \ifnum #1\expandafter \@firstoftwo
 \else \expandafter \@secondoftwo
 \fi
}%
\providecommand \@ifx [1]{%
 \ifx #1\expandafter \@firstoftwo
 \else \expandafter \@secondoftwo
 \fi
}%
\providecommand \natexlab [1]{#1}%
\providecommand \enquote  [1]{``#1''}%
\providecommand \bibnamefont  [1]{#1}%
\providecommand \bibfnamefont [1]{#1}%
\providecommand \citenamefont [1]{#1}%
\providecommand \href@noop [0]{\@secondoftwo}%
\providecommand \href [0]{\begingroup \@sanitize@url \@href}%
\providecommand \@href[1]{\@@startlink{#1}\@@href}%
\providecommand \@@href[1]{\endgroup#1\@@endlink}%
\providecommand \@sanitize@url [0]{\catcode `\\12\catcode `\$12\catcode `\&12\catcode `\#12\catcode `\^12\catcode `\_12\catcode `\%12\relax}%
\providecommand \@@startlink[1]{}%
\providecommand \@@endlink[0]{}%
\providecommand \url  [0]{\begingroup\@sanitize@url \@url }%
\providecommand \@url [1]{\endgroup\@href {#1}{\urlprefix }}%
\providecommand \urlprefix  [0]{URL }%
\providecommand \Eprint [0]{\href }%
\providecommand \doibase [0]{https://doi.org/}%
\providecommand \selectlanguage [0]{\@gobble}%
\providecommand \bibinfo  [0]{\@secondoftwo}%
\providecommand \bibfield  [0]{\@secondoftwo}%
\providecommand \translation [1]{[#1]}%
\providecommand \BibitemOpen [0]{}%
\providecommand \bibitemStop [0]{}%
\providecommand \bibitemNoStop [0]{.\EOS\space}%
\providecommand \EOS [0]{\spacefactor3000\relax}%
\providecommand \BibitemShut  [1]{\csname bibitem#1\endcsname}%
\let\auto@bib@innerbib\@empty
\bibitem [{\citenamefont {Kitaev}(2003)}]{Kitaev_2003}%
  \BibitemOpen
  \bibfield  {author} {\bibinfo {author} {\bibfnamefont {A.}~\bibnamefont {Kitaev}},\ }\href {https://doi.org/10.1016/s0003-4916(02)00018-0} {\bibfield  {journal} {\bibinfo  {journal} {Annals of Physics}\ }\textbf {\bibinfo {volume} {303}},\ \bibinfo {pages} {2–30} (\bibinfo {year} {2003})}\BibitemShut {NoStop}%
\bibitem [{\citenamefont {Bombin}\ and\ \citenamefont {Martin-Delgado}(2006)}]{Bombin_2006}%
  \BibitemOpen
  \bibfield  {author} {\bibinfo {author} {\bibfnamefont {H.}~\bibnamefont {Bombin}}\ and\ \bibinfo {author} {\bibfnamefont {M.~A.}\ \bibnamefont {Martin-Delgado}},\ }\href {https://doi.org/10.1103/PhysRevLett.97.180501} {\bibfield  {journal} {\bibinfo  {journal} {Phys. Rev. Lett.}\ }\textbf {\bibinfo {volume} {97}},\ \bibinfo {pages} {180501} (\bibinfo {year} {2006})}\BibitemShut {NoStop}%
\bibitem [{\citenamefont {Dennis}\ \emph {et~al.}(2002)\citenamefont {Dennis}, \citenamefont {Kitaev}, \citenamefont {Landahl},\ and\ \citenamefont {Preskill}}]{Dennis_2002}%
  \BibitemOpen
  \bibfield  {author} {\bibinfo {author} {\bibfnamefont {E.}~\bibnamefont {Dennis}}, \bibinfo {author} {\bibfnamefont {A.}~\bibnamefont {Kitaev}}, \bibinfo {author} {\bibfnamefont {A.}~\bibnamefont {Landahl}},\ and\ \bibinfo {author} {\bibfnamefont {J.}~\bibnamefont {Preskill}},\ }\href {https://doi.org/10.1063/1.1499754} {\bibfield  {journal} {\bibinfo  {journal} {Journal of Mathematical Physics}\ }\textbf {\bibinfo {volume} {43}},\ \bibinfo {pages} {4452–4505} (\bibinfo {year} {2002})}\BibitemShut {NoStop}%
\bibitem [{\citenamefont {Horsman}\ \emph {et~al.}(2012)\citenamefont {Horsman}, \citenamefont {Fowler}, \citenamefont {Devitt},\ and\ \citenamefont {Meter}}]{Horsman_2012}%
  \BibitemOpen
  \bibfield  {author} {\bibinfo {author} {\bibfnamefont {D.}~\bibnamefont {Horsman}}, \bibinfo {author} {\bibfnamefont {A.~G.}\ \bibnamefont {Fowler}}, \bibinfo {author} {\bibfnamefont {S.}~\bibnamefont {Devitt}},\ and\ \bibinfo {author} {\bibfnamefont {R.~V.}\ \bibnamefont {Meter}},\ }\href {https://doi.org/10.1088/1367-2630/14/12/123011} {\bibfield  {journal} {\bibinfo  {journal} {New Journal of Physics}\ }\textbf {\bibinfo {volume} {14}},\ \bibinfo {pages} {123011} (\bibinfo {year} {2012})}\BibitemShut {NoStop}%
\bibitem [{\citenamefont {Fowler}\ \emph {et~al.}(2012)\citenamefont {Fowler}, \citenamefont {Mariantoni}, \citenamefont {Martinis},\ and\ \citenamefont {Cleland}}]{Fowler_2012}%
  \BibitemOpen
  \bibfield  {author} {\bibinfo {author} {\bibfnamefont {A.~G.}\ \bibnamefont {Fowler}}, \bibinfo {author} {\bibfnamefont {M.}~\bibnamefont {Mariantoni}}, \bibinfo {author} {\bibfnamefont {J.~M.}\ \bibnamefont {Martinis}},\ and\ \bibinfo {author} {\bibfnamefont {A.~N.}\ \bibnamefont {Cleland}},\ }\href {https://doi.org/10.1103/PhysRevA.86.032324} {\bibfield  {journal} {\bibinfo  {journal} {Phys. Rev. A}\ }\textbf {\bibinfo {volume} {86}},\ \bibinfo {pages} {032324} (\bibinfo {year} {2012})}\BibitemShut {NoStop}%
\bibitem [{\citenamefont {Landahl}\ and\ \citenamefont {Ryan-Anderson}(2014)}]{landahl2014}%
  \BibitemOpen
  \bibfield  {author} {\bibinfo {author} {\bibfnamefont {A.~J.}\ \bibnamefont {Landahl}}\ and\ \bibinfo {author} {\bibfnamefont {C.}~\bibnamefont {Ryan-Anderson}},\ }\href {https://arxiv.org/abs/1407.5103} {\bibinfo {title} {Quantum computing by color-code lattice surgery}} (\bibinfo {year} {2014}),\ \Eprint {https://arxiv.org/abs/1407.5103} {arXiv:1407.5103 [quant-ph]} \BibitemShut {NoStop}%
\bibitem [{\citenamefont {Brown}\ \emph {et~al.}(2017)\citenamefont {Brown}, \citenamefont {Laubscher}, \citenamefont {Kesselring},\ and\ \citenamefont {Wootton}}]{Brown_2017}%
  \BibitemOpen
  \bibfield  {author} {\bibinfo {author} {\bibfnamefont {B.~J.}\ \bibnamefont {Brown}}, \bibinfo {author} {\bibfnamefont {K.}~\bibnamefont {Laubscher}}, \bibinfo {author} {\bibfnamefont {M.~S.}\ \bibnamefont {Kesselring}},\ and\ \bibinfo {author} {\bibfnamefont {J.~R.}\ \bibnamefont {Wootton}},\ }\href {https://doi.org/10.1103/PhysRevX.7.021029} {\bibfield  {journal} {\bibinfo  {journal} {Phys. Rev. X}\ }\textbf {\bibinfo {volume} {7}},\ \bibinfo {pages} {021029} (\bibinfo {year} {2017})}\BibitemShut {NoStop}%
\bibitem [{\citenamefont {Zhou}\ \emph {et~al.}(2025)\citenamefont {Zhou}, \citenamefont {Zhao}, \citenamefont {Cain}, \citenamefont {Bluvstein}, \citenamefont {Maskara}, \citenamefont {Duckering}, \citenamefont {Hu}, \citenamefont {Wang}, \citenamefont {Kubica},\ and\ \citenamefont {Lukin}}]{zhou2024alg}%
  \BibitemOpen
  \bibfield  {author} {\bibinfo {author} {\bibfnamefont {H.}~\bibnamefont {Zhou}}, \bibinfo {author} {\bibfnamefont {C.}~\bibnamefont {Zhao}}, \bibinfo {author} {\bibfnamefont {M.}~\bibnamefont {Cain}}, \bibinfo {author} {\bibfnamefont {D.}~\bibnamefont {Bluvstein}}, \bibinfo {author} {\bibfnamefont {N.}~\bibnamefont {Maskara}}, \bibinfo {author} {\bibfnamefont {C.}~\bibnamefont {Duckering}}, \bibinfo {author} {\bibfnamefont {H.-Y.}\ \bibnamefont {Hu}}, \bibinfo {author} {\bibfnamefont {S.-T.}\ \bibnamefont {Wang}}, \bibinfo {author} {\bibfnamefont {A.}~\bibnamefont {Kubica}},\ and\ \bibinfo {author} {\bibfnamefont {M.~D.}\ \bibnamefont {Lukin}},\ }\href {https://arxiv.org/abs/2406.17653} {\bibinfo {title} {Low-overhead transversal fault tolerance for universal quantum computation}} (\bibinfo {year} {2025}),\ \Eprint {https://arxiv.org/abs/2406.17653} {arXiv:2406.17653 [quant-ph]} \BibitemShut {NoStop}%
\bibitem [{\citenamefont {{C. Ryan-Anderson et al.}}(2022)}]{ryananderson2022}%
  \BibitemOpen
  \bibfield  {author} {\bibinfo {author} {\bibnamefont {{C. Ryan-Anderson et al.}}},\ }\href {https://arxiv.org/abs/2208.01863} {\bibinfo {title} {Implementing fault-tolerant entangling gates on the five-qubit code and the color code}} (\bibinfo {year} {2022}),\ \Eprint {https://arxiv.org/abs/2208.01863} {arXiv:2208.01863 [quant-ph]} \BibitemShut {NoStop}%
\bibitem [{\citenamefont {{Google Quantum AI}}(2023)}]{Google_2023}%
  \BibitemOpen
  \bibfield  {author} {\bibinfo {author} {\bibnamefont {{Google Quantum AI}}},\ }\href {https://doi.org/10.1038/s41586-022-05434-1} {\bibfield  {journal} {\bibinfo  {journal} {Nature}\ }\textbf {\bibinfo {volume} {614}},\ \bibinfo {pages} {676} (\bibinfo {year} {2023})}\BibitemShut {NoStop}%
\bibitem [{\citenamefont {{Dolev Bluvstein et al.}}(2023)}]{Bluvstein_2023}%
  \BibitemOpen
  \bibfield  {author} {\bibinfo {author} {\bibnamefont {{Dolev Bluvstein et al.}}},\ }\href {https://doi.org/10.1038/s41586-023-06927-3} {\bibfield  {journal} {\bibinfo  {journal} {Nature}\ }\textbf {\bibinfo {volume} {626}},\ \bibinfo {pages} {58–65} (\bibinfo {year} {2023})}\BibitemShut {NoStop}%
\bibitem [{\citenamefont {{Google Quantum AI et al.}}(2024)}]{Google_2024}%
  \BibitemOpen
  \bibfield  {author} {\bibinfo {author} {\bibnamefont {{Google Quantum AI et al.}}},\ }\href {https://arxiv.org/abs/2408.13687} {\bibinfo {title} {Quantum error correction below the surface code threshold}} (\bibinfo {year} {2024}),\ \Eprint {https://arxiv.org/abs/2408.13687} {arXiv:2408.13687 [quant-ph]} \BibitemShut {NoStop}%
\bibitem [{\citenamefont {Bravyi}\ and\ \citenamefont {Terhal}(2009)}]{Bravyi_2009}%
  \BibitemOpen
  \bibfield  {author} {\bibinfo {author} {\bibfnamefont {S.}~\bibnamefont {Bravyi}}\ and\ \bibinfo {author} {\bibfnamefont {B.}~\bibnamefont {Terhal}},\ }\href {https://doi.org/10.1088/1367-2630/11/4/043029} {\bibfield  {journal} {\bibinfo  {journal} {New Journal of Physics}\ }\textbf {\bibinfo {volume} {11}},\ \bibinfo {pages} {043029} (\bibinfo {year} {2009})}\BibitemShut {NoStop}%
\bibitem [{\citenamefont {Bravyi}\ \emph {et~al.}(2010)\citenamefont {Bravyi}, \citenamefont {Poulin},\ and\ \citenamefont {Terhal}}]{BPT_2010}%
  \BibitemOpen
  \bibfield  {author} {\bibinfo {author} {\bibfnamefont {S.}~\bibnamefont {Bravyi}}, \bibinfo {author} {\bibfnamefont {D.}~\bibnamefont {Poulin}},\ and\ \bibinfo {author} {\bibfnamefont {B.}~\bibnamefont {Terhal}},\ }\href {https://doi.org/10.1103/PhysRevLett.104.050503} {\bibfield  {journal} {\bibinfo  {journal} {Phys. Rev. Lett.}\ }\textbf {\bibinfo {volume} {104}},\ \bibinfo {pages} {050503} (\bibinfo {year} {2010})}\BibitemShut {NoStop}%
\bibitem [{\citenamefont {Bravyi}\ and\ \citenamefont {K\"onig}(2013)}]{Bravyi_2013_transversal}%
  \BibitemOpen
  \bibfield  {author} {\bibinfo {author} {\bibfnamefont {S.}~\bibnamefont {Bravyi}}\ and\ \bibinfo {author} {\bibfnamefont {R.}~\bibnamefont {K\"onig}},\ }\href {https://doi.org/10.1103/PhysRevLett.110.170503} {\bibfield  {journal} {\bibinfo  {journal} {Phys. Rev. Lett.}\ }\textbf {\bibinfo {volume} {110}},\ \bibinfo {pages} {170503} (\bibinfo {year} {2013})}\BibitemShut {NoStop}%
\bibitem [{\citenamefont {Bombin}\ and\ \citenamefont {Martin-Delgado}(2007{\natexlab{a}})}]{Bombin_2007_3D}%
  \BibitemOpen
  \bibfield  {author} {\bibinfo {author} {\bibfnamefont {H.}~\bibnamefont {Bombin}}\ and\ \bibinfo {author} {\bibfnamefont {M.~A.}\ \bibnamefont {Martin-Delgado}},\ }\href {https://doi.org/10.1103/PhysRevLett.98.160502} {\bibfield  {journal} {\bibinfo  {journal} {Phys. Rev. Lett.}\ }\textbf {\bibinfo {volume} {98}},\ \bibinfo {pages} {160502} (\bibinfo {year} {2007}{\natexlab{a}})}\BibitemShut {NoStop}%
\bibitem [{\citenamefont {Kubica}\ \emph {et~al.}(2015)\citenamefont {Kubica}, \citenamefont {Yoshida},\ and\ \citenamefont {Pastawski}}]{Kubica_2015}%
  \BibitemOpen
  \bibfield  {author} {\bibinfo {author} {\bibfnamefont {A.}~\bibnamefont {Kubica}}, \bibinfo {author} {\bibfnamefont {B.}~\bibnamefont {Yoshida}},\ and\ \bibinfo {author} {\bibfnamefont {F.}~\bibnamefont {Pastawski}},\ }\href {https://doi.org/10.1088/1367-2630/17/8/083026} {\bibfield  {journal} {\bibinfo  {journal} {New Journal of Physics}\ }\textbf {\bibinfo {volume} {17}},\ \bibinfo {pages} {083026} (\bibinfo {year} {2015})}\BibitemShut {NoStop}%
\bibitem [{\citenamefont {Bomb\'{\i}n}(2015)}]{Bombin_2015}%
  \BibitemOpen
  \bibfield  {author} {\bibinfo {author} {\bibfnamefont {H.}~\bibnamefont {Bomb\'{\i}n}},\ }\href {https://doi.org/10.1103/PhysRevX.5.031043} {\bibfield  {journal} {\bibinfo  {journal} {Phys. Rev. X}\ }\textbf {\bibinfo {volume} {5}},\ \bibinfo {pages} {031043} (\bibinfo {year} {2015})}\BibitemShut {NoStop}%
\bibitem [{\citenamefont {Kubica}\ and\ \citenamefont {Vasmer}(2022)}]{Kubica_2022}%
  \BibitemOpen
  \bibfield  {author} {\bibinfo {author} {\bibfnamefont {A.}~\bibnamefont {Kubica}}\ and\ \bibinfo {author} {\bibfnamefont {M.}~\bibnamefont {Vasmer}},\ }\href {https://doi.org/10.1038/s41467-022-33923-4} {\bibfield  {journal} {\bibinfo  {journal} {Nature Communications}\ }\textbf {\bibinfo {volume} {13}},\ \bibinfo {pages} {6272} (\bibinfo {year} {2022})}\BibitemShut {NoStop}%
\bibitem [{\citenamefont {Stahl}(2024)}]{Stahl_2024}%
  \BibitemOpen
  \bibfield  {author} {\bibinfo {author} {\bibfnamefont {C.}~\bibnamefont {Stahl}},\ }\href {https://doi.org/10.1103/PhysRevB.110.075143} {\bibfield  {journal} {\bibinfo  {journal} {Phys. Rev. B}\ }\textbf {\bibinfo {volume} {110}},\ \bibinfo {pages} {075143} (\bibinfo {year} {2024})}\BibitemShut {NoStop}%
\bibitem [{\citenamefont {Gottesman}(1997)}]{gottesman1997stabilizer}%
  \BibitemOpen
  \bibfield  {author} {\bibinfo {author} {\bibfnamefont {D.}~\bibnamefont {Gottesman}},\ }\href {https://arxiv.org/abs/quant-ph/9705052} {\bibinfo {title} {Stabilizer codes and quantum error correction}} (\bibinfo {year} {1997}),\ \Eprint {https://arxiv.org/abs/quant-ph/9705052} {arXiv:quant-ph/9705052 [quant-ph]} \BibitemShut {NoStop}%
\bibitem [{\citenamefont {Breuckmann}\ and\ \citenamefont {Eberhardt}(2021)}]{Breuckmann_2021}%
  \BibitemOpen
  \bibfield  {author} {\bibinfo {author} {\bibfnamefont {N.~P.}\ \bibnamefont {Breuckmann}}\ and\ \bibinfo {author} {\bibfnamefont {J.~N.}\ \bibnamefont {Eberhardt}},\ }\href {https://doi.org/10.1103/PRXQuantum.2.040101} {\bibfield  {journal} {\bibinfo  {journal} {PRX Quantum}\ }\textbf {\bibinfo {volume} {2}},\ \bibinfo {pages} {040101} (\bibinfo {year} {2021})}\BibitemShut {NoStop}%
\bibitem [{\citenamefont {Freedman}\ \emph {et~al.}(2002)\citenamefont {Freedman}, \citenamefont {Meyer},\ and\ \citenamefont {Luo}}]{Freedman_2002}%
  \BibitemOpen
  \bibfield  {author} {\bibinfo {author} {\bibfnamefont {M.}~\bibnamefont {Freedman}}, \bibinfo {author} {\bibfnamefont {D.}~\bibnamefont {Meyer}},\ and\ \bibinfo {author} {\bibfnamefont {F.}~\bibnamefont {Luo}},\ }\bibinfo {title} {Z2-systolic freedom and quantum codes},\ in\ \href@noop {} {\emph {\bibinfo {booktitle} {Mathematics of Quantum Computation}}}\ (\bibinfo  {publisher} {CRC Press},\ \bibinfo {year} {2002})\ pp.\ \bibinfo {pages} {287--320}\BibitemShut {NoStop}%
\bibitem [{\citenamefont {Tillich}\ and\ \citenamefont {Zemor}(2014)}]{Tillich_2014_HGP}%
  \BibitemOpen
  \bibfield  {author} {\bibinfo {author} {\bibfnamefont {J.-P.}\ \bibnamefont {Tillich}}\ and\ \bibinfo {author} {\bibfnamefont {G.}~\bibnamefont {Zemor}},\ }\href {https://doi.org/10.1109/tit.2013.2292061} {\bibfield  {journal} {\bibinfo  {journal} {IEEE Transactions on Information Theory}\ }\textbf {\bibinfo {volume} {60}},\ \bibinfo {pages} {1193–1202} (\bibinfo {year} {2014})}\BibitemShut {NoStop}%
\bibitem [{\citenamefont {Panteleev}\ and\ \citenamefont {Kalachev}(2022)}]{Panteleev_2022_LDPC}%
  \BibitemOpen
  \bibfield  {author} {\bibinfo {author} {\bibfnamefont {P.}~\bibnamefont {Panteleev}}\ and\ \bibinfo {author} {\bibfnamefont {G.}~\bibnamefont {Kalachev}},\ }in\ \href {https://doi.org/10.1145/3519935.3520017} {\emph {\bibinfo {booktitle} {Proceedings of the 54th Annual ACM SIGACT Symposium on Theory of Computing}}},\ \bibinfo {series and number} {STOC 2022}\ (\bibinfo  {publisher} {Association for Computing Machinery},\ \bibinfo {address} {New York, NY, USA},\ \bibinfo {year} {2022})\ p.\ \bibinfo {pages} {375–388}\BibitemShut {NoStop}%
\bibitem [{\citenamefont {Leverrier}\ and\ \citenamefont {Zemor}(2022)}]{qTanner_codes}%
  \BibitemOpen
  \bibfield  {author} {\bibinfo {author} {\bibfnamefont {A.}~\bibnamefont {Leverrier}}\ and\ \bibinfo {author} {\bibfnamefont {G.}~\bibnamefont {Zemor}},\ }in\ \href {https://doi.org/10.1109/FOCS54457.2022.00117} {\emph {\bibinfo {booktitle} {2022 IEEE 63rd Annual Symposium on Foundations of Computer Science (FOCS)}}}\ (\bibinfo  {publisher} {IEEE Computer Society},\ \bibinfo {address} {Los Alamitos, CA, USA},\ \bibinfo {year} {2022})\ pp.\ \bibinfo {pages} {872--883}\BibitemShut {NoStop}%
\bibitem [{\citenamefont {Dinur}\ \emph {et~al.}(2023)\citenamefont {Dinur}, \citenamefont {Hsieh}, \citenamefont {Lin},\ and\ \citenamefont {Vidick}}]{Dinur_2023_LDPC}%
  \BibitemOpen
  \bibfield  {author} {\bibinfo {author} {\bibfnamefont {I.}~\bibnamefont {Dinur}}, \bibinfo {author} {\bibfnamefont {M.-H.}\ \bibnamefont {Hsieh}}, \bibinfo {author} {\bibfnamefont {T.-C.}\ \bibnamefont {Lin}},\ and\ \bibinfo {author} {\bibfnamefont {T.}~\bibnamefont {Vidick}},\ }in\ \href {https://doi.org/10.1145/3564246.3585101} {\emph {\bibinfo {booktitle} {Proceedings of the 55th Annual ACM Symposium on Theory of Computing}}},\ \bibinfo {series and number} {STOC 2023}\ (\bibinfo  {publisher} {Association for Computing Machinery},\ \bibinfo {address} {New York, NY, USA},\ \bibinfo {year} {2023})\ p.\ \bibinfo {pages} {905–918}\BibitemShut {NoStop}%
\bibitem [{\citenamefont {Baspin}\ and\ \citenamefont {Krishna}(2022)}]{Baspin_2022}%
  \BibitemOpen
  \bibfield  {author} {\bibinfo {author} {\bibfnamefont {N.}~\bibnamefont {Baspin}}\ and\ \bibinfo {author} {\bibfnamefont {A.}~\bibnamefont {Krishna}},\ }\href {https://doi.org/10.22331/q-2022-05-13-711} {\bibfield  {journal} {\bibinfo  {journal} {Quantum}\ }\textbf {\bibinfo {volume} {6}},\ \bibinfo {pages} {711} (\bibinfo {year} {2022})}\BibitemShut {NoStop}%
\bibitem [{\citenamefont {Baspin}\ \emph {et~al.}(2024)\citenamefont {Baspin}, \citenamefont {Guruswami}, \citenamefont {Krishna},\ and\ \citenamefont {Li}}]{baspin2023improved}%
  \BibitemOpen
  \bibfield  {author} {\bibinfo {author} {\bibfnamefont {N.}~\bibnamefont {Baspin}}, \bibinfo {author} {\bibfnamefont {V.}~\bibnamefont {Guruswami}}, \bibinfo {author} {\bibfnamefont {A.}~\bibnamefont {Krishna}},\ and\ \bibinfo {author} {\bibfnamefont {R.}~\bibnamefont {Li}},\ }\href {https://doi.org/10.1088/2058-9565/ad8370} {\bibfield  {journal} {\bibinfo  {journal} {Quantum Science and Technology}\ }\textbf {\bibinfo {volume} {10}},\ \bibinfo {pages} {015021} (\bibinfo {year} {2024})}\BibitemShut {NoStop}%
\bibitem [{\citenamefont {Dai}\ and\ \citenamefont {Li}(2025)}]{dai2024local}%
  \BibitemOpen
  \bibfield  {author} {\bibinfo {author} {\bibfnamefont {S.}~\bibnamefont {Dai}}\ and\ \bibinfo {author} {\bibfnamefont {R.}~\bibnamefont {Li}},\ }in\ \href {https://doi.org/10.1145/3717823.3718113} {\emph {\bibinfo {booktitle} {Proceedings of the 57th Annual ACM Symposium on Theory of Computing}}},\ \bibinfo {series and number} {STOC '25}\ (\bibinfo  {publisher} {Association for Computing Machinery},\ \bibinfo {address} {New York, NY, USA},\ \bibinfo {year} {2025})\ p.\ \bibinfo {pages} {677–688}\BibitemShut {NoStop}%
\bibitem [{\citenamefont {Cirac}\ and\ \citenamefont {Zoller}(1995)}]{Cirac_1995}%
  \BibitemOpen
  \bibfield  {author} {\bibinfo {author} {\bibfnamefont {J.~I.}\ \bibnamefont {Cirac}}\ and\ \bibinfo {author} {\bibfnamefont {P.}~\bibnamefont {Zoller}},\ }\href {https://doi.org/10.1103/PhysRevLett.74.4091} {\bibfield  {journal} {\bibinfo  {journal} {Phys. Rev. Lett.}\ }\textbf {\bibinfo {volume} {74}},\ \bibinfo {pages} {4091} (\bibinfo {year} {1995})}\BibitemShut {NoStop}%
\bibitem [{\citenamefont {Chen}\ \emph {et~al.}(2024)\citenamefont {Chen}, \citenamefont {Nielsen}, \citenamefont {Ebert}, \citenamefont {Inlek}, \citenamefont {Wright}, \citenamefont {Chaplin}, \citenamefont {Maksymov}, \citenamefont {Páez}, \citenamefont {Poudel}, \citenamefont {Maunz},\ and\ \citenamefont {Gamble}}]{chen2023ionq}%
  \BibitemOpen
  \bibfield  {author} {\bibinfo {author} {\bibfnamefont {J.-S.}\ \bibnamefont {Chen}}, \bibinfo {author} {\bibfnamefont {E.}~\bibnamefont {Nielsen}}, \bibinfo {author} {\bibfnamefont {M.}~\bibnamefont {Ebert}}, \bibinfo {author} {\bibfnamefont {V.}~\bibnamefont {Inlek}}, \bibinfo {author} {\bibfnamefont {K.}~\bibnamefont {Wright}}, \bibinfo {author} {\bibfnamefont {V.}~\bibnamefont {Chaplin}}, \bibinfo {author} {\bibfnamefont {A.}~\bibnamefont {Maksymov}}, \bibinfo {author} {\bibfnamefont {E.}~\bibnamefont {Páez}}, \bibinfo {author} {\bibfnamefont {A.}~\bibnamefont {Poudel}}, \bibinfo {author} {\bibfnamefont {P.}~\bibnamefont {Maunz}},\ and\ \bibinfo {author} {\bibfnamefont {J.}~\bibnamefont {Gamble}},\ }\href {https://doi.org/10.22331/q-2024-11-07-1516} {\bibfield  {journal} {\bibinfo  {journal} {Quantum}\ }\textbf {\bibinfo {volume} {8}},\ \bibinfo {pages} {1516} (\bibinfo {year} {2024})}\BibitemShut {NoStop}%
\bibitem [{\citenamefont {et~al.}(2023)}]{Quantinuum_H2-1}%
  \BibitemOpen
  \bibfield  {author} {\bibinfo {author} {\bibfnamefont {S.~A.~M.}\ \bibnamefont {et~al.}},\ }\href {https://doi.org/10.1103/PhysRevX.13.041052} {\bibfield  {journal} {\bibinfo  {journal} {Phys. Rev. X}\ }\textbf {\bibinfo {volume} {13}},\ \bibinfo {pages} {041052} (\bibinfo {year} {2023})}\BibitemShut {NoStop}%
\bibitem [{\citenamefont {Saffman}(2016)}]{Saffman_2016}%
  \BibitemOpen
  \bibfield  {author} {\bibinfo {author} {\bibfnamefont {M.}~\bibnamefont {Saffman}},\ }\href {https://doi.org/10.1088/0953-4075/49/20/202001} {\bibfield  {journal} {\bibinfo  {journal} {Journal of Physics B: Atomic, Molecular and Optical Physics}\ }\textbf {\bibinfo {volume} {49}},\ \bibinfo {pages} {202001} (\bibinfo {year} {2016})}\BibitemShut {NoStop}%
\bibitem [{\citenamefont {Jenkins}\ \emph {et~al.}(2022)\citenamefont {Jenkins}, \citenamefont {Lis}, \citenamefont {Senoo}, \citenamefont {McGrew},\ and\ \citenamefont {Kaufman}}]{Jenkins_2022}%
  \BibitemOpen
  \bibfield  {author} {\bibinfo {author} {\bibfnamefont {A.}~\bibnamefont {Jenkins}}, \bibinfo {author} {\bibfnamefont {J.~W.}\ \bibnamefont {Lis}}, \bibinfo {author} {\bibfnamefont {A.}~\bibnamefont {Senoo}}, \bibinfo {author} {\bibfnamefont {W.~F.}\ \bibnamefont {McGrew}},\ and\ \bibinfo {author} {\bibfnamefont {A.~M.}\ \bibnamefont {Kaufman}},\ }\href {https://doi.org/10.1103/PhysRevX.12.021027} {\bibfield  {journal} {\bibinfo  {journal} {Phys. Rev. X}\ }\textbf {\bibinfo {volume} {12}},\ \bibinfo {pages} {021027} (\bibinfo {year} {2022})}\BibitemShut {NoStop}%
\bibitem [{\citenamefont {Evered}\ \emph {et~al.}(2023)\citenamefont {Evered}, \citenamefont {Bluvstein}, \citenamefont {Kalinowski}, \citenamefont {Ebadi}, \citenamefont {Manovitz}, \citenamefont {Zhou}, \citenamefont {Li}, \citenamefont {Geim}, \citenamefont {Wang}, \citenamefont {Maskara}, \citenamefont {Levine}, \citenamefont {Semeghini}, \citenamefont {Greiner}, \citenamefont {Vuletić},\ and\ \citenamefont {Lukin}}]{Evered_2023}%
  \BibitemOpen
  \bibfield  {author} {\bibinfo {author} {\bibfnamefont {S.~J.}\ \bibnamefont {Evered}}, \bibinfo {author} {\bibfnamefont {D.}~\bibnamefont {Bluvstein}}, \bibinfo {author} {\bibfnamefont {M.}~\bibnamefont {Kalinowski}}, \bibinfo {author} {\bibfnamefont {S.}~\bibnamefont {Ebadi}}, \bibinfo {author} {\bibfnamefont {T.}~\bibnamefont {Manovitz}}, \bibinfo {author} {\bibfnamefont {H.}~\bibnamefont {Zhou}}, \bibinfo {author} {\bibfnamefont {S.~H.}\ \bibnamefont {Li}}, \bibinfo {author} {\bibfnamefont {A.~A.}\ \bibnamefont {Geim}}, \bibinfo {author} {\bibfnamefont {T.~T.}\ \bibnamefont {Wang}}, \bibinfo {author} {\bibfnamefont {N.}~\bibnamefont {Maskara}}, \bibinfo {author} {\bibfnamefont {H.}~\bibnamefont {Levine}}, \bibinfo {author} {\bibfnamefont {G.}~\bibnamefont {Semeghini}}, \bibinfo {author} {\bibfnamefont {M.}~\bibnamefont {Greiner}}, \bibinfo {author} {\bibfnamefont {V.}~\bibnamefont {Vuletić}},\ and\ \bibinfo {author} {\bibfnamefont {M.~D.}\ \bibnamefont {Lukin}},\ }\href
  {https://doi.org/10.1038/s41586-023-06481-y} {\bibfield  {journal} {\bibinfo  {journal} {Nature}\ }\textbf {\bibinfo {volume} {622}},\ \bibinfo {pages} {268–272} (\bibinfo {year} {2023})}\BibitemShut {NoStop}%
\bibitem [{\citenamefont {Hong}\ \emph {et~al.}(2024{\natexlab{a}})\citenamefont {Hong}, \citenamefont {Durso-Sabina}, \citenamefont {Hayes},\ and\ \citenamefont {Lucas}}]{hong2024ghz}%
  \BibitemOpen
  \bibfield  {author} {\bibinfo {author} {\bibfnamefont {Y.}~\bibnamefont {Hong}}, \bibinfo {author} {\bibfnamefont {E.}~\bibnamefont {Durso-Sabina}}, \bibinfo {author} {\bibfnamefont {D.}~\bibnamefont {Hayes}},\ and\ \bibinfo {author} {\bibfnamefont {A.}~\bibnamefont {Lucas}},\ }\href {https://doi.org/10.1103/PhysRevLett.133.180601} {\bibfield  {journal} {\bibinfo  {journal} {Phys. Rev. Lett.}\ }\textbf {\bibinfo {volume} {133}},\ \bibinfo {pages} {180601} (\bibinfo {year} {2024}{\natexlab{a}})}\BibitemShut {NoStop}%
\bibitem [{\citenamefont {Berthusen}\ \emph {et~al.}(2024)\citenamefont {Berthusen}, \citenamefont {Dreiling}, \citenamefont {Foltz}, \citenamefont {Gaebler}, \citenamefont {Gatterman}, \citenamefont {Gresh}, \citenamefont {Hewitt}, \citenamefont {Mills}, \citenamefont {Moses}, \citenamefont {Neyenhuis}, \citenamefont {Siegfried},\ and\ \citenamefont {Hayes}}]{berthusen2024ss}%
  \BibitemOpen
  \bibfield  {author} {\bibinfo {author} {\bibfnamefont {N.}~\bibnamefont {Berthusen}}, \bibinfo {author} {\bibfnamefont {J.}~\bibnamefont {Dreiling}}, \bibinfo {author} {\bibfnamefont {C.}~\bibnamefont {Foltz}}, \bibinfo {author} {\bibfnamefont {J.~P.}\ \bibnamefont {Gaebler}}, \bibinfo {author} {\bibfnamefont {T.~M.}\ \bibnamefont {Gatterman}}, \bibinfo {author} {\bibfnamefont {D.}~\bibnamefont {Gresh}}, \bibinfo {author} {\bibfnamefont {N.}~\bibnamefont {Hewitt}}, \bibinfo {author} {\bibfnamefont {M.}~\bibnamefont {Mills}}, \bibinfo {author} {\bibfnamefont {S.~A.}\ \bibnamefont {Moses}}, \bibinfo {author} {\bibfnamefont {B.}~\bibnamefont {Neyenhuis}}, \bibinfo {author} {\bibfnamefont {P.}~\bibnamefont {Siegfried}},\ and\ \bibinfo {author} {\bibfnamefont {D.}~\bibnamefont {Hayes}},\ }\href {https://doi.org/10.1103/PhysRevA.110.062413} {\bibfield  {journal} {\bibinfo  {journal} {Phys. Rev. A}\ }\textbf {\bibinfo {volume} {110}},\ \bibinfo {pages} {062413} (\bibinfo {year} {2024})}\BibitemShut {NoStop}%
\bibitem [{\citenamefont {Reichardt}\ \emph {et~al.}(2024)\citenamefont {Reichardt}, \citenamefont {Aasen}, \citenamefont {Chao}, \citenamefont {Chernoguzov}, \citenamefont {van Dam}, \citenamefont {Gaebler}, \citenamefont {Gresh}, \citenamefont {Lucchetti}, \citenamefont {Mills}, \citenamefont {Moses}, \citenamefont {Neyenhuis}, \citenamefont {Paetznick}, \citenamefont {Paz}, \citenamefont {Siegfried}, \citenamefont {da~Silva}, \citenamefont {Svore}, \citenamefont {Wang},\ and\ \citenamefont {Zanner}}]{reichardt2024tess}%
  \BibitemOpen
  \bibfield  {author} {\bibinfo {author} {\bibfnamefont {B.~W.}\ \bibnamefont {Reichardt}}, \bibinfo {author} {\bibfnamefont {D.}~\bibnamefont {Aasen}}, \bibinfo {author} {\bibfnamefont {R.}~\bibnamefont {Chao}}, \bibinfo {author} {\bibfnamefont {A.}~\bibnamefont {Chernoguzov}}, \bibinfo {author} {\bibfnamefont {W.}~\bibnamefont {van Dam}}, \bibinfo {author} {\bibfnamefont {J.~P.}\ \bibnamefont {Gaebler}}, \bibinfo {author} {\bibfnamefont {D.}~\bibnamefont {Gresh}}, \bibinfo {author} {\bibfnamefont {D.}~\bibnamefont {Lucchetti}}, \bibinfo {author} {\bibfnamefont {M.}~\bibnamefont {Mills}}, \bibinfo {author} {\bibfnamefont {S.~A.}\ \bibnamefont {Moses}}, \bibinfo {author} {\bibfnamefont {B.}~\bibnamefont {Neyenhuis}}, \bibinfo {author} {\bibfnamefont {A.}~\bibnamefont {Paetznick}}, \bibinfo {author} {\bibfnamefont {A.}~\bibnamefont {Paz}}, \bibinfo {author} {\bibfnamefont {P.~E.}\ \bibnamefont {Siegfried}}, \bibinfo {author} {\bibfnamefont {M.~P.}\ \bibnamefont {da~Silva}}, \bibinfo {author} {\bibfnamefont
  {K.~M.}\ \bibnamefont {Svore}}, \bibinfo {author} {\bibfnamefont {Z.}~\bibnamefont {Wang}},\ and\ \bibinfo {author} {\bibfnamefont {M.}~\bibnamefont {Zanner}},\ }\href {https://arxiv.org/abs/2409.04628} {\bibinfo {title} {Demonstration of quantum computation and error correction with a tesseract code}} (\bibinfo {year} {2024}),\ \Eprint {https://arxiv.org/abs/2409.04628} {arXiv:2409.04628 [quant-ph]} \BibitemShut {NoStop}%
\bibitem [{\citenamefont {Fawzi}\ \emph {et~al.}(2018)\citenamefont {Fawzi}, \citenamefont {Grospellier},\ and\ \citenamefont {Leverrier}}]{Fawzi_2018}%
  \BibitemOpen
  \bibfield  {author} {\bibinfo {author} {\bibfnamefont {O.}~\bibnamefont {Fawzi}}, \bibinfo {author} {\bibfnamefont {A.}~\bibnamefont {Grospellier}},\ and\ \bibinfo {author} {\bibfnamefont {A.}~\bibnamefont {Leverrier}},\ }in\ \href {https://doi.org/10.1109/focs.2018.00076} {\emph {\bibinfo {booktitle} {2018 IEEE 59th Annual Symposium on Foundations of Computer Science (FOCS)}}},\ Vol.~\bibinfo {volume} {14}\ (\bibinfo  {publisher} {IEEE},\ \bibinfo {year} {2018})\ p.\ \bibinfo {pages} {743–754}\BibitemShut {NoStop}%
\bibitem [{\citenamefont {Hong}\ \emph {et~al.}(2025)\citenamefont {Hong}, \citenamefont {Guo},\ and\ \citenamefont {Lucas}}]{hong2024thermal}%
  \BibitemOpen
  \bibfield  {author} {\bibinfo {author} {\bibfnamefont {Y.}~\bibnamefont {Hong}}, \bibinfo {author} {\bibfnamefont {J.}~\bibnamefont {Guo}},\ and\ \bibinfo {author} {\bibfnamefont {A.}~\bibnamefont {Lucas}},\ }\bibfield  {journal} {\bibinfo  {journal} {Nature Communications}\ }\textbf {\bibinfo {volume} {16}},\ \href {https://doi.org/10.1038/s41467-024-55570-7} {10.1038/s41467-024-55570-7} (\bibinfo {year} {2025})\BibitemShut {NoStop}%
\bibitem [{\citenamefont {Manes}\ and\ \citenamefont {Claes}(2025)}]{manes2023hgp}%
  \BibitemOpen
  \bibfield  {author} {\bibinfo {author} {\bibfnamefont {A.~G.}\ \bibnamefont {Manes}}\ and\ \bibinfo {author} {\bibfnamefont {J.}~\bibnamefont {Claes}},\ }\href {https://doi.org/10.22331/q-2025-01-30-1618} {\bibfield  {journal} {\bibinfo  {journal} {Quantum}\ }\textbf {\bibinfo {volume} {9}},\ \bibinfo {pages} {1618} (\bibinfo {year} {2025})}\BibitemShut {NoStop}%
\bibitem [{\citenamefont {Krishna}\ and\ \citenamefont {Poulin}(2021)}]{Krishna_2021}%
  \BibitemOpen
  \bibfield  {author} {\bibinfo {author} {\bibfnamefont {A.}~\bibnamefont {Krishna}}\ and\ \bibinfo {author} {\bibfnamefont {D.}~\bibnamefont {Poulin}},\ }\href {https://doi.org/10.1103/PhysRevX.11.011023} {\bibfield  {journal} {\bibinfo  {journal} {Phys. Rev. X}\ }\textbf {\bibinfo {volume} {11}},\ \bibinfo {pages} {011023} (\bibinfo {year} {2021})}\BibitemShut {NoStop}%
\bibitem [{\citenamefont {Breuckmann}\ and\ \citenamefont {Burton}(2024)}]{Breuckmann_2024_fold}%
  \BibitemOpen
  \bibfield  {author} {\bibinfo {author} {\bibfnamefont {N.~P.}\ \bibnamefont {Breuckmann}}\ and\ \bibinfo {author} {\bibfnamefont {S.}~\bibnamefont {Burton}},\ }\href {https://doi.org/10.22331/q-2024-06-13-1372} {\bibfield  {journal} {\bibinfo  {journal} {Quantum}\ }\textbf {\bibinfo {volume} {8}},\ \bibinfo {pages} {1372} (\bibinfo {year} {2024})}\BibitemShut {NoStop}%
\bibitem [{\citenamefont {Quintavalle}\ \emph {et~al.}(2023)\citenamefont {Quintavalle}, \citenamefont {Webster},\ and\ \citenamefont {Vasmer}}]{Quintavalle_2023}%
  \BibitemOpen
  \bibfield  {author} {\bibinfo {author} {\bibfnamefont {A.~O.}\ \bibnamefont {Quintavalle}}, \bibinfo {author} {\bibfnamefont {P.}~\bibnamefont {Webster}},\ and\ \bibinfo {author} {\bibfnamefont {M.}~\bibnamefont {Vasmer}},\ }\href {https://doi.org/10.22331/q-2023-10-24-1153} {\bibfield  {journal} {\bibinfo  {journal} {Quantum}\ }\textbf {\bibinfo {volume} {7}},\ \bibinfo {pages} {1153} (\bibinfo {year} {2023})}\BibitemShut {NoStop}%
\bibitem [{\citenamefont {Hong}\ \emph {et~al.}(2024{\natexlab{b}})\citenamefont {Hong}, \citenamefont {Marinelli}, \citenamefont {Kaufman},\ and\ \citenamefont {Lucas}}]{LRESC}%
  \BibitemOpen
  \bibfield  {author} {\bibinfo {author} {\bibfnamefont {Y.}~\bibnamefont {Hong}}, \bibinfo {author} {\bibfnamefont {M.}~\bibnamefont {Marinelli}}, \bibinfo {author} {\bibfnamefont {A.~M.}\ \bibnamefont {Kaufman}},\ and\ \bibinfo {author} {\bibfnamefont {A.}~\bibnamefont {Lucas}},\ }\href {https://doi.org/10.1103/PhysRevA.110.022607} {\bibfield  {journal} {\bibinfo  {journal} {Phys. Rev. A}\ }\textbf {\bibinfo {volume} {110}},\ \bibinfo {pages} {022607} (\bibinfo {year} {2024}{\natexlab{b}})}\BibitemShut {NoStop}%
\bibitem [{\citenamefont {Gottesman}(2014)}]{gottesman2014}%
  \BibitemOpen
  \bibfield  {author} {\bibinfo {author} {\bibfnamefont {D.}~\bibnamefont {Gottesman}},\ }\href@noop {} {\bibfield  {journal} {\bibinfo  {journal} {Quantum Info. Comput.}\ }\textbf {\bibinfo {volume} {14}},\ \bibinfo {pages} {1338–1372} (\bibinfo {year} {2014})}\BibitemShut {NoStop}%
\bibitem [{\citenamefont {Xu}\ \emph {et~al.}(2024)\citenamefont {Xu}, \citenamefont {Bonilla~Ataides}, \citenamefont {Pattison}, \citenamefont {Raveendran}, \citenamefont {Bluvstein}, \citenamefont {Wurtz}, \citenamefont {Vasi{\'c}}, \citenamefont {Lukin}, \citenamefont {Jiang},\ and\ \citenamefont {Zhou}}]{Xu_2024_constant}%
  \BibitemOpen
  \bibfield  {author} {\bibinfo {author} {\bibfnamefont {Q.}~\bibnamefont {Xu}}, \bibinfo {author} {\bibfnamefont {J.~P.}\ \bibnamefont {Bonilla~Ataides}}, \bibinfo {author} {\bibfnamefont {C.~A.}\ \bibnamefont {Pattison}}, \bibinfo {author} {\bibfnamefont {N.}~\bibnamefont {Raveendran}}, \bibinfo {author} {\bibfnamefont {D.}~\bibnamefont {Bluvstein}}, \bibinfo {author} {\bibfnamefont {J.}~\bibnamefont {Wurtz}}, \bibinfo {author} {\bibfnamefont {B.}~\bibnamefont {Vasi{\'c}}}, \bibinfo {author} {\bibfnamefont {M.~D.}\ \bibnamefont {Lukin}}, \bibinfo {author} {\bibfnamefont {L.}~\bibnamefont {Jiang}},\ and\ \bibinfo {author} {\bibfnamefont {H.}~\bibnamefont {Zhou}},\ }\href {https://doi.org/10.1038/s41567-024-02479-z} {\bibfield  {journal} {\bibinfo  {journal} {Nature Physics}\ }\textbf {\bibinfo {volume} {20}},\ \bibinfo {pages} {1084} (\bibinfo {year} {2024})}\BibitemShut {NoStop}%
\bibitem [{\citenamefont {Bombín}(2016)}]{Bombin_2016}%
  \BibitemOpen
  \bibfield  {author} {\bibinfo {author} {\bibfnamefont {H.}~\bibnamefont {Bombín}},\ }\href {https://doi.org/10.1088/1367-2630/18/4/043038} {\bibfield  {journal} {\bibinfo  {journal} {New Journal of Physics}\ }\textbf {\bibinfo {volume} {18}},\ \bibinfo {pages} {043038} (\bibinfo {year} {2016})}\BibitemShut {NoStop}%
\bibitem [{\citenamefont {Brown}(2020)}]{Brown_2020}%
  \BibitemOpen
  \bibfield  {author} {\bibinfo {author} {\bibfnamefont {B.~J.}\ \bibnamefont {Brown}},\ }\href {http://dx.doi.org/10.1126/sciadv.aay4929} {\bibfield  {journal} {\bibinfo  {journal} {Science Advances}\ }\textbf {\bibinfo {volume} {6}} (\bibinfo {year} {2020})}\BibitemShut {NoStop}%
\bibitem [{\citenamefont {Scruby}\ \emph {et~al.}(2022{\natexlab{a}})\citenamefont {Scruby}, \citenamefont {Browne}, \citenamefont {Webster},\ and\ \citenamefont {Vasmer}}]{Scruby_2022_jit}%
  \BibitemOpen
  \bibfield  {author} {\bibinfo {author} {\bibfnamefont {T.~R.}\ \bibnamefont {Scruby}}, \bibinfo {author} {\bibfnamefont {D.~E.}\ \bibnamefont {Browne}}, \bibinfo {author} {\bibfnamefont {P.}~\bibnamefont {Webster}},\ and\ \bibinfo {author} {\bibfnamefont {M.}~\bibnamefont {Vasmer}},\ }\href {https://doi.org/10.22331/q-2022-05-24-721} {\bibfield  {journal} {\bibinfo  {journal} {Quantum}\ }\textbf {\bibinfo {volume} {6}},\ \bibinfo {pages} {721} (\bibinfo {year} {2022}{\natexlab{a}})}\BibitemShut {NoStop}%
\bibitem [{\citenamefont {Raussendorf}\ \emph {et~al.}(2005)\citenamefont {Raussendorf}, \citenamefont {Bravyi},\ and\ \citenamefont {Harrington}}]{3D_cluster}%
  \BibitemOpen
  \bibfield  {author} {\bibinfo {author} {\bibfnamefont {R.}~\bibnamefont {Raussendorf}}, \bibinfo {author} {\bibfnamefont {S.}~\bibnamefont {Bravyi}},\ and\ \bibinfo {author} {\bibfnamefont {J.}~\bibnamefont {Harrington}},\ }\href {https://doi.org/10.1103/PhysRevA.71.062313} {\bibfield  {journal} {\bibinfo  {journal} {Phys. Rev. A}\ }\textbf {\bibinfo {volume} {71}},\ \bibinfo {pages} {062313} (\bibinfo {year} {2005})}\BibitemShut {NoStop}%
\bibitem [{\citenamefont {Bravyi}\ \emph {et~al.}(2020)\citenamefont {Bravyi}, \citenamefont {Gosset}, \citenamefont {König},\ and\ \citenamefont {Tomamichel}}]{Bravyi_2020_sc}%
  \BibitemOpen
  \bibfield  {author} {\bibinfo {author} {\bibfnamefont {S.}~\bibnamefont {Bravyi}}, \bibinfo {author} {\bibfnamefont {D.}~\bibnamefont {Gosset}}, \bibinfo {author} {\bibfnamefont {R.}~\bibnamefont {König}},\ and\ \bibinfo {author} {\bibfnamefont {M.}~\bibnamefont {Tomamichel}},\ }\href {https://doi.org/10.1038/s41567-020-0948-z} {\bibfield  {journal} {\bibinfo  {journal} {Nature Physics}\ }\textbf {\bibinfo {volume} {16}},\ \bibinfo {pages} {1040–1045} (\bibinfo {year} {2020})}\BibitemShut {NoStop}%
\bibitem [{\citenamefont {Scruby}\ \emph {et~al.}(2022{\natexlab{b}})\citenamefont {Scruby}, \citenamefont {Vasmer},\ and\ \citenamefont {Browne}}]{Scruby_2022}%
  \BibitemOpen
  \bibfield  {author} {\bibinfo {author} {\bibfnamefont {T.~R.}\ \bibnamefont {Scruby}}, \bibinfo {author} {\bibfnamefont {M.}~\bibnamefont {Vasmer}},\ and\ \bibinfo {author} {\bibfnamefont {D.~E.}\ \bibnamefont {Browne}},\ }\href {https://doi.org/10.1103/PhysRevResearch.4.043052} {\bibfield  {journal} {\bibinfo  {journal} {Phys. Rev. Res.}\ }\textbf {\bibinfo {volume} {4}},\ \bibinfo {pages} {043052} (\bibinfo {year} {2022}{\natexlab{b}})}\BibitemShut {NoStop}%
\bibitem [{\citenamefont {Higgott}\ and\ \citenamefont {Breuckmann}(2023)}]{Higgott_2023}%
  \BibitemOpen
  \bibfield  {author} {\bibinfo {author} {\bibfnamefont {O.}~\bibnamefont {Higgott}}\ and\ \bibinfo {author} {\bibfnamefont {N.~P.}\ \bibnamefont {Breuckmann}},\ }\href {https://doi.org/10.1103/PRXQuantum.4.020332} {\bibfield  {journal} {\bibinfo  {journal} {PRX Quantum}\ }\textbf {\bibinfo {volume} {4}},\ \bibinfo {pages} {020332} (\bibinfo {year} {2023})}\BibitemShut {NoStop}%
\bibitem [{\citenamefont {Gallager}(1962)}]{Gallager_1962}%
  \BibitemOpen
  \bibfield  {author} {\bibinfo {author} {\bibfnamefont {R.}~\bibnamefont {Gallager}},\ }\href {https://doi.org/10.1109/TIT.1962.1057683} {\bibfield  {journal} {\bibinfo  {journal} {IRE Transactions on Information Theory}\ }\textbf {\bibinfo {volume} {8}},\ \bibinfo {pages} {21} (\bibinfo {year} {1962})}\BibitemShut {NoStop}%
\bibitem [{\citenamefont {Bombin}\ and\ \citenamefont {Martin-Delgado}(2007{\natexlab{b}})}]{Bombin_2007_Hom}%
  \BibitemOpen
  \bibfield  {author} {\bibinfo {author} {\bibfnamefont {H.}~\bibnamefont {Bombin}}\ and\ \bibinfo {author} {\bibfnamefont {M.~A.}\ \bibnamefont {Martin-Delgado}},\ }\href {http://dx.doi.org/10.1063/1.2731356} {\bibfield  {journal} {\bibinfo  {journal} {Journal of Mathematical Physics}\ }\textbf {\bibinfo {volume} {48}} (\bibinfo {year} {2007}{\natexlab{b}})}\BibitemShut {NoStop}%
\bibitem [{\citenamefont {Calderbank}\ and\ \citenamefont {Shor}(1996)}]{Calderbank_1996}%
  \BibitemOpen
  \bibfield  {author} {\bibinfo {author} {\bibfnamefont {A.~R.}\ \bibnamefont {Calderbank}}\ and\ \bibinfo {author} {\bibfnamefont {P.~W.}\ \bibnamefont {Shor}},\ }\href {https://doi.org/10.1103/PhysRevA.54.1098} {\bibfield  {journal} {\bibinfo  {journal} {Phys. Rev. A}\ }\textbf {\bibinfo {volume} {54}},\ \bibinfo {pages} {1098} (\bibinfo {year} {1996})}\BibitemShut {NoStop}%
\bibitem [{\citenamefont {Steane}(1996)}]{Steane_1996}%
  \BibitemOpen
  \bibfield  {author} {\bibinfo {author} {\bibfnamefont {A.}~\bibnamefont {Steane}},\ }\href {https://doi.org/10.1098/rspa.1996.0136} {\bibfield  {journal} {\bibinfo  {journal} {Proceedings of the Royal Society of London. Series A: Mathematical, Physical and Engineering Sciences}\ }\textbf {\bibinfo {volume} {452}},\ \bibinfo {pages} {2551} (\bibinfo {year} {1996})}\BibitemShut {NoStop}%
\bibitem [{\citenamefont {Freedman}\ and\ \citenamefont {Hastings}(2014)}]{Freedman_2014}%
  \BibitemOpen
  \bibfield  {author} {\bibinfo {author} {\bibfnamefont {M.~H.}\ \bibnamefont {Freedman}}\ and\ \bibinfo {author} {\bibfnamefont {M.~B.}\ \bibnamefont {Hastings}},\ }\href@noop {} {\bibfield  {journal} {\bibinfo  {journal} {Quantum Info. Comput.}\ }\textbf {\bibinfo {volume} {14}},\ \bibinfo {pages} {144–180} (\bibinfo {year} {2014})}\BibitemShut {NoStop}%
\bibitem [{\citenamefont {Bravyi}\ and\ \citenamefont {Hastings}(2014)}]{bravyi2013hom}%
  \BibitemOpen
  \bibfield  {author} {\bibinfo {author} {\bibfnamefont {S.}~\bibnamefont {Bravyi}}\ and\ \bibinfo {author} {\bibfnamefont {M.~B.}\ \bibnamefont {Hastings}},\ }in\ \href {https://doi.org/10.1145/2591796.2591870} {\emph {\bibinfo {booktitle} {Proceedings of the Forty-Sixth Annual ACM Symposium on Theory of Computing}}},\ \bibinfo {series and number} {STOC '14}\ (\bibinfo  {publisher} {Association for Computing Machinery},\ \bibinfo {address} {New York, NY, USA},\ \bibinfo {year} {2014})\ p.\ \bibinfo {pages} {273–282}\BibitemShut {NoStop}%
\bibitem [{\citenamefont {Mosheiff}\ \emph {et~al.}(2021)\citenamefont {Mosheiff}, \citenamefont {Resch}, \citenamefont {Ron-Zewi}, \citenamefont {Silas},\ and\ \citenamefont {Wootters}}]{Mosheiff_2021}%
  \BibitemOpen
  \bibfield  {author} {\bibinfo {author} {\bibfnamefont {J.}~\bibnamefont {Mosheiff}}, \bibinfo {author} {\bibfnamefont {N.}~\bibnamefont {Resch}}, \bibinfo {author} {\bibfnamefont {N.}~\bibnamefont {Ron-Zewi}}, \bibinfo {author} {\bibfnamefont {S.}~\bibnamefont {Silas}},\ and\ \bibinfo {author} {\bibfnamefont {M.}~\bibnamefont {Wootters}},\ }\href {https://doi.org/10.1137/20m1365934} {\bibfield  {journal} {\bibinfo  {journal} {SIAM Journal on Computing}\ ,\ \bibinfo {pages} {FOCS20}} (\bibinfo {year} {2021})}\BibitemShut {NoStop}%
\bibitem [{\citenamefont {Shor}(1995)}]{Shor_1995}%
  \BibitemOpen
  \bibfield  {author} {\bibinfo {author} {\bibfnamefont {P.~W.}\ \bibnamefont {Shor}},\ }\href {https://doi.org/10.1103/PhysRevA.52.R2493} {\bibfield  {journal} {\bibinfo  {journal} {Phys. Rev. A}\ }\textbf {\bibinfo {volume} {52}},\ \bibinfo {pages} {R2493} (\bibinfo {year} {1995})}\BibitemShut {NoStop}%
\bibitem [{\citenamefont {Campbell}(2019)}]{Campbell_2019}%
  \BibitemOpen
  \bibfield  {author} {\bibinfo {author} {\bibfnamefont {E.~T.}\ \bibnamefont {Campbell}},\ }\href {https://doi.org/10.1088/2058-9565/aafc8f} {\bibfield  {journal} {\bibinfo  {journal} {Quantum Science and Technology}\ }\textbf {\bibinfo {volume} {4}},\ \bibinfo {pages} {025006} (\bibinfo {year} {2019})}\BibitemShut {NoStop}%
\bibitem [{\citenamefont {Quintavalle}\ \emph {et~al.}(2021)\citenamefont {Quintavalle}, \citenamefont {Vasmer}, \citenamefont {Roffe},\ and\ \citenamefont {Campbell}}]{Quintavalle_2021}%
  \BibitemOpen
  \bibfield  {author} {\bibinfo {author} {\bibfnamefont {A.~O.}\ \bibnamefont {Quintavalle}}, \bibinfo {author} {\bibfnamefont {M.}~\bibnamefont {Vasmer}}, \bibinfo {author} {\bibfnamefont {J.}~\bibnamefont {Roffe}},\ and\ \bibinfo {author} {\bibfnamefont {E.~T.}\ \bibnamefont {Campbell}},\ }\href {https://doi.org/10.1103/PRXQuantum.2.020340} {\bibfield  {journal} {\bibinfo  {journal} {PRX Quantum}\ }\textbf {\bibinfo {volume} {2}},\ \bibinfo {pages} {020340} (\bibinfo {year} {2021})}\BibitemShut {NoStop}%
\bibitem [{\citenamefont {Anshu}\ \emph {et~al.}(2023)\citenamefont {Anshu}, \citenamefont {Breuckmann},\ and\ \citenamefont {Nirkhe}}]{Anshu_2023}%
  \BibitemOpen
  \bibfield  {author} {\bibinfo {author} {\bibfnamefont {A.}~\bibnamefont {Anshu}}, \bibinfo {author} {\bibfnamefont {N.~P.}\ \bibnamefont {Breuckmann}},\ and\ \bibinfo {author} {\bibfnamefont {C.}~\bibnamefont {Nirkhe}},\ }in\ \href {https://doi.org/10.1145/3564246.3585114} {\emph {\bibinfo {booktitle} {Proceedings of the 55th Annual ACM Symposium on Theory of Computing}}},\ \bibinfo {series} {STOC ’23}, Vol.~\bibinfo {volume} {5}\ (\bibinfo  {publisher} {ACM},\ \bibinfo {year} {2023})\ p.\ \bibinfo {pages} {1090–1096}\BibitemShut {NoStop}%
\bibitem [{\citenamefont {Ben-Sasson}\ \emph {et~al.}(2009)\citenamefont {Ben-Sasson}, \citenamefont {Guruswami}, \citenamefont {Kaufman}, \citenamefont {Sudan},\ and\ \citenamefont {Viderman}}]{Sasson_2009}%
  \BibitemOpen
  \bibfield  {author} {\bibinfo {author} {\bibfnamefont {E.}~\bibnamefont {Ben-Sasson}}, \bibinfo {author} {\bibfnamefont {V.}~\bibnamefont {Guruswami}}, \bibinfo {author} {\bibfnamefont {T.}~\bibnamefont {Kaufman}}, \bibinfo {author} {\bibfnamefont {M.}~\bibnamefont {Sudan}},\ and\ \bibinfo {author} {\bibfnamefont {M.}~\bibnamefont {Viderman}},\ }in\ \href {https://doi.org/10.1109/CCC.2009.6} {\emph {\bibinfo {booktitle} {2009 24th Annual IEEE Conference on Computational Complexity}}}\ (\bibinfo {year} {2009})\ pp.\ \bibinfo {pages} {52--61}\BibitemShut {NoStop}%
\bibitem [{\citenamefont {Leverrier}\ \emph {et~al.}(2015)\citenamefont {Leverrier}, \citenamefont {Tillich},\ and\ \citenamefont {Zemor}}]{Leverrier_2015}%
  \BibitemOpen
  \bibfield  {author} {\bibinfo {author} {\bibfnamefont {A.}~\bibnamefont {Leverrier}}, \bibinfo {author} {\bibfnamefont {J.-P.}\ \bibnamefont {Tillich}},\ and\ \bibinfo {author} {\bibfnamefont {G.}~\bibnamefont {Zemor}},\ }in\ \href {https://doi.org/10.1109/focs.2015.55} {\emph {\bibinfo {booktitle} {2015 IEEE 56th Annual Symposium on Foundations of Computer Science}}}\ (\bibinfo  {publisher} {IEEE},\ \bibinfo {year} {2015})\BibitemShut {NoStop}%
\bibitem [{\citenamefont {Xu}\ \emph {et~al.}(2025)\citenamefont {Xu}, \citenamefont {Zhou}, \citenamefont {Zheng}, \citenamefont {Bluvstein}, \citenamefont {Ataides}, \citenamefont {Lukin},\ and\ \citenamefont {Jiang}}]{xu2024fastparallel}%
  \BibitemOpen
  \bibfield  {author} {\bibinfo {author} {\bibfnamefont {Q.}~\bibnamefont {Xu}}, \bibinfo {author} {\bibfnamefont {H.}~\bibnamefont {Zhou}}, \bibinfo {author} {\bibfnamefont {G.}~\bibnamefont {Zheng}}, \bibinfo {author} {\bibfnamefont {D.}~\bibnamefont {Bluvstein}}, \bibinfo {author} {\bibfnamefont {J.~P.~B.}\ \bibnamefont {Ataides}}, \bibinfo {author} {\bibfnamefont {M.~D.}\ \bibnamefont {Lukin}},\ and\ \bibinfo {author} {\bibfnamefont {L.}~\bibnamefont {Jiang}},\ }\href {https://doi.org/10.1103/PhysRevX.15.021065} {\bibfield  {journal} {\bibinfo  {journal} {Phys. Rev. X}\ }\textbf {\bibinfo {volume} {15}},\ \bibinfo {pages} {021065} (\bibinfo {year} {2025})}\BibitemShut {NoStop}%
\bibitem [{\citenamefont {Richardson}\ and\ \citenamefont {Urbanke}(2008)}]{ModernCodingTheory}%
  \BibitemOpen
  \bibfield  {author} {\bibinfo {author} {\bibfnamefont {T.}~\bibnamefont {Richardson}}\ and\ \bibinfo {author} {\bibfnamefont {R.}~\bibnamefont {Urbanke}},\ }\href@noop {} {\emph {\bibinfo {title} {Modern coding theory}}}\ (\bibinfo  {publisher} {Cambridge University Press},\ \bibinfo {year} {2008})\BibitemShut {NoStop}%
\bibitem [{\citenamefont {Hastings}(2017)}]{Hastings_2017_weight}%
  \BibitemOpen
  \bibfield  {author} {\bibinfo {author} {\bibfnamefont {M.~B.}\ \bibnamefont {Hastings}},\ }\href@noop {} {\bibfield  {journal} {\bibinfo  {journal} {Quantum Info. Comput.}\ }\textbf {\bibinfo {volume} {17}},\ \bibinfo {pages} {1307–1334} (\bibinfo {year} {2017})}\BibitemShut {NoStop}%
\bibitem [{\citenamefont {Evra}\ \emph {et~al.}(2020)\citenamefont {Evra}, \citenamefont {Kaufman},\ and\ \citenamefont {Z{\'e}mor}}]{Evra_2020}%
  \BibitemOpen
  \bibfield  {author} {\bibinfo {author} {\bibfnamefont {S.}~\bibnamefont {Evra}}, \bibinfo {author} {\bibfnamefont {T.}~\bibnamefont {Kaufman}},\ and\ \bibinfo {author} {\bibfnamefont {G.}~\bibnamefont {Z{\'e}mor}},\ }\href {https://api.semanticscholar.org/CorpusID:231684733} {\bibfield  {journal} {\bibinfo  {journal} {2020 IEEE 61st Annual Symposium on Foundations of Computer Science (FOCS)}\ ,\ \bibinfo {pages} {218}} (\bibinfo {year} {2020})}\BibitemShut {NoStop}%
\bibitem [{\citenamefont {Zeng}\ and\ \citenamefont {Pryadko}(2019)}]{Zeng_2019}%
  \BibitemOpen
  \bibfield  {author} {\bibinfo {author} {\bibfnamefont {W.}~\bibnamefont {Zeng}}\ and\ \bibinfo {author} {\bibfnamefont {L.~P.}\ \bibnamefont {Pryadko}},\ }\href {https://doi.org/10.1103/PhysRevLett.122.230501} {\bibfield  {journal} {\bibinfo  {journal} {Phys. Rev. Lett.}\ }\textbf {\bibinfo {volume} {122}},\ \bibinfo {pages} {230501} (\bibinfo {year} {2019})}\BibitemShut {NoStop}%
\bibitem [{\citenamefont {Bolt}\ \emph {et~al.}(2016)\citenamefont {Bolt}, \citenamefont {Duclos-Cianci}, \citenamefont {Poulin},\ and\ \citenamefont {Stace}}]{Bolt_2016}%
  \BibitemOpen
  \bibfield  {author} {\bibinfo {author} {\bibfnamefont {A.}~\bibnamefont {Bolt}}, \bibinfo {author} {\bibfnamefont {G.}~\bibnamefont {Duclos-Cianci}}, \bibinfo {author} {\bibfnamefont {D.}~\bibnamefont {Poulin}},\ and\ \bibinfo {author} {\bibfnamefont {T.~M.}\ \bibnamefont {Stace}},\ }\href {https://doi.org/10.1103/PhysRevLett.117.070501} {\bibfield  {journal} {\bibinfo  {journal} {Phys. Rev. Lett.}\ }\textbf {\bibinfo {volume} {117}},\ \bibinfo {pages} {070501} (\bibinfo {year} {2016})}\BibitemShut {NoStop}%
\bibitem [{\citenamefont {Bolt}\ \emph {et~al.}(2018)\citenamefont {Bolt}, \citenamefont {Poulin},\ and\ \citenamefont {Stace}}]{Bolt_2018}%
  \BibitemOpen
  \bibfield  {author} {\bibinfo {author} {\bibfnamefont {A.}~\bibnamefont {Bolt}}, \bibinfo {author} {\bibfnamefont {D.}~\bibnamefont {Poulin}},\ and\ \bibinfo {author} {\bibfnamefont {T.~M.}\ \bibnamefont {Stace}},\ }\href {https://doi.org/10.1103/PhysRevA.98.062302} {\bibfield  {journal} {\bibinfo  {journal} {Phys. Rev. A}\ }\textbf {\bibinfo {volume} {98}},\ \bibinfo {pages} {062302} (\bibinfo {year} {2018})}\BibitemShut {NoStop}%
\bibitem [{\citenamefont {Bacon}\ \emph {et~al.}(2015)\citenamefont {Bacon}, \citenamefont {Flammia}, \citenamefont {Harrow},\ and\ \citenamefont {Shi}}]{Bacon_2015}%
  \BibitemOpen
  \bibfield  {author} {\bibinfo {author} {\bibfnamefont {D.}~\bibnamefont {Bacon}}, \bibinfo {author} {\bibfnamefont {S.~T.}\ \bibnamefont {Flammia}}, \bibinfo {author} {\bibfnamefont {A.~W.}\ \bibnamefont {Harrow}},\ and\ \bibinfo {author} {\bibfnamefont {J.}~\bibnamefont {Shi}},\ }in\ \href {https://doi.org/10.1145/2746539.2746608} {\emph {\bibinfo {booktitle} {Proceedings of the Forty-Seventh Annual ACM Symposium on Theory of Computing}}},\ \bibinfo {series and number} {STOC '15}\ (\bibinfo  {publisher} {Association for Computing Machinery},\ \bibinfo {address} {New York, NY, USA},\ \bibinfo {year} {2015})\ p.\ \bibinfo {pages} {327–334}\BibitemShut {NoStop}%
\bibitem [{\citenamefont {Gottesman}(2022)}]{gottesman2022}%
  \BibitemOpen
  \bibfield  {author} {\bibinfo {author} {\bibfnamefont {D.}~\bibnamefont {Gottesman}},\ }\href {https://arxiv.org/abs/2210.15844} {\bibinfo {title} {Opportunities and challenges in fault-tolerant quantum computation}} (\bibinfo {year} {2022}),\ \Eprint {https://arxiv.org/abs/2210.15844} {arXiv:2210.15844 [quant-ph]} \BibitemShut {NoStop}%
\bibitem [{\citenamefont {Delfosse}\ and\ \citenamefont {Paetznick}(2023)}]{delfosse2023}%
  \BibitemOpen
  \bibfield  {author} {\bibinfo {author} {\bibfnamefont {N.}~\bibnamefont {Delfosse}}\ and\ \bibinfo {author} {\bibfnamefont {A.}~\bibnamefont {Paetznick}},\ }\href {https://arxiv.org/abs/2304.05943} {\bibinfo {title} {Spacetime codes of clifford circuits}} (\bibinfo {year} {2023}),\ \Eprint {https://arxiv.org/abs/2304.05943} {arXiv:2304.05943 [quant-ph]} \BibitemShut {NoStop}%
\bibitem [{\citenamefont {Sabo}\ \emph {et~al.}(2024)\citenamefont {Sabo}, \citenamefont {Gunderman}, \citenamefont {Ide}, \citenamefont {Vasmer},\ and\ \citenamefont {Dauphinais}}]{sabo2024weight}%
  \BibitemOpen
  \bibfield  {author} {\bibinfo {author} {\bibfnamefont {E.}~\bibnamefont {Sabo}}, \bibinfo {author} {\bibfnamefont {L.~G.}\ \bibnamefont {Gunderman}}, \bibinfo {author} {\bibfnamefont {B.}~\bibnamefont {Ide}}, \bibinfo {author} {\bibfnamefont {M.}~\bibnamefont {Vasmer}},\ and\ \bibinfo {author} {\bibfnamefont {G.}~\bibnamefont {Dauphinais}},\ }\href {https://doi.org/10.1103/PRXQuantum.5.040302} {\bibfield  {journal} {\bibinfo  {journal} {PRX Quantum}\ }\textbf {\bibinfo {volume} {5}},\ \bibinfo {pages} {040302} (\bibinfo {year} {2024})}\BibitemShut {NoStop}%
\bibitem [{\citenamefont {Etzion}\ \emph {et~al.}(1999)\citenamefont {Etzion}, \citenamefont {Trachtenberg},\ and\ \citenamefont {Vardy}}]{Etzion_1999}%
  \BibitemOpen
  \bibfield  {author} {\bibinfo {author} {\bibfnamefont {T.}~\bibnamefont {Etzion}}, \bibinfo {author} {\bibfnamefont {A.}~\bibnamefont {Trachtenberg}},\ and\ \bibinfo {author} {\bibfnamefont {A.}~\bibnamefont {Vardy}},\ }\href {https://doi.org/10.1109/18.782170} {\bibfield  {journal} {\bibinfo  {journal} {IEEE Transactions on Information Theory}\ }\textbf {\bibinfo {volume} {45}},\ \bibinfo {pages} {2173} (\bibinfo {year} {1999})}\BibitemShut {NoStop}%
\bibitem [{\citenamefont {Berlekamp}\ \emph {et~al.}(1978)\citenamefont {Berlekamp}, \citenamefont {McEliece},\ and\ \citenamefont {van Tilborg}}]{Berlekamp_1978}%
  \BibitemOpen
  \bibfield  {author} {\bibinfo {author} {\bibfnamefont {E.}~\bibnamefont {Berlekamp}}, \bibinfo {author} {\bibfnamefont {R.}~\bibnamefont {McEliece}},\ and\ \bibinfo {author} {\bibfnamefont {H.}~\bibnamefont {van Tilborg}},\ }\href {https://doi.org/10.1109/TIT.1978.1055873} {\bibfield  {journal} {\bibinfo  {journal} {IEEE Transactions on Information Theory}\ }\textbf {\bibinfo {volume} {24}},\ \bibinfo {pages} {384} (\bibinfo {year} {1978})}\BibitemShut {NoStop}%
\bibitem [{\citenamefont {Roffe}(2022)}]{Roffe_LDPC_Python_tools_2022}%
  \BibitemOpen
  \bibfield  {author} {\bibinfo {author} {\bibfnamefont {J.}~\bibnamefont {Roffe}},\ }\href {https://pypi.org/project/ldpc/} {\bibinfo {title} {{LDPC: Python tools for low density parity check codes}}} (\bibinfo {year} {2022})\BibitemShut {NoStop}%
\bibitem [{\citenamefont {Panteleev}\ and\ \citenamefont {Kalachev}(2021)}]{Panteleev_BPOSD_2021}%
  \BibitemOpen
  \bibfield  {author} {\bibinfo {author} {\bibfnamefont {P.}~\bibnamefont {Panteleev}}\ and\ \bibinfo {author} {\bibfnamefont {G.}~\bibnamefont {Kalachev}},\ }\href {https://doi.org/10.22331/q-2021-11-22-585} {\bibfield  {journal} {\bibinfo  {journal} {Quantum}\ }\textbf {\bibinfo {volume} {5}},\ \bibinfo {pages} {585} (\bibinfo {year} {2021})}\BibitemShut {NoStop}%
\bibitem [{\citenamefont {Roffe}\ \emph {et~al.}(2020)\citenamefont {Roffe}, \citenamefont {White}, \citenamefont {Burton},\ and\ \citenamefont {Campbell}}]{Roffe_decoding_2020}%
  \BibitemOpen
  \bibfield  {author} {\bibinfo {author} {\bibfnamefont {J.}~\bibnamefont {Roffe}}, \bibinfo {author} {\bibfnamefont {D.~R.}\ \bibnamefont {White}}, \bibinfo {author} {\bibfnamefont {S.}~\bibnamefont {Burton}},\ and\ \bibinfo {author} {\bibfnamefont {E.}~\bibnamefont {Campbell}},\ }\href {https://doi.org/10.1103/PhysRevResearch.2.043423} {\bibfield  {journal} {\bibinfo  {journal} {Phys. Rev. Res.}\ }\textbf {\bibinfo {volume} {2}},\ \bibinfo {pages} {043423} (\bibinfo {year} {2020})}\BibitemShut {NoStop}%
\bibitem [{\citenamefont {Wolanski}\ and\ \citenamefont {Barber}(2024)}]{wolanski2024}%
  \BibitemOpen
  \bibfield  {author} {\bibinfo {author} {\bibfnamefont {S.}~\bibnamefont {Wolanski}}\ and\ \bibinfo {author} {\bibfnamefont {B.}~\bibnamefont {Barber}},\ }\href {https://arxiv.org/abs/2406.14527} {\bibinfo {title} {Ambiguity clustering: an accurate and efficient decoder for qldpc codes}} (\bibinfo {year} {2024}),\ \Eprint {https://arxiv.org/abs/2406.14527} {arXiv:2406.14527 [quant-ph]} \BibitemShut {NoStop}%
\bibitem [{\citenamefont {Hillmann}\ \emph {et~al.}(2025)\citenamefont {Hillmann}, \citenamefont {Berent}, \citenamefont {Quintavalle}, \citenamefont {Eisert}, \citenamefont {Wille},\ and\ \citenamefont {Roffe}}]{BPLSD}%
  \BibitemOpen
  \bibfield  {author} {\bibinfo {author} {\bibfnamefont {T.}~\bibnamefont {Hillmann}}, \bibinfo {author} {\bibfnamefont {L.}~\bibnamefont {Berent}}, \bibinfo {author} {\bibfnamefont {A.~O.}\ \bibnamefont {Quintavalle}}, \bibinfo {author} {\bibfnamefont {J.}~\bibnamefont {Eisert}}, \bibinfo {author} {\bibfnamefont {R.}~\bibnamefont {Wille}},\ and\ \bibinfo {author} {\bibfnamefont {J.}~\bibnamefont {Roffe}},\ }\bibfield  {journal} {\bibinfo  {journal} {Nature Communications}\ }\textbf {\bibinfo {volume} {16}},\ \href {https://doi.org/10.1038/s41467-025-63214-7} {10.1038/s41467-025-63214-7} (\bibinfo {year} {2025})\BibitemShut {NoStop}%
\bibitem [{\citenamefont {deMarti iOlius}\ \emph {et~al.}(2024)\citenamefont {deMarti iOlius}, \citenamefont {Martinez}, \citenamefont {Roffe},\ and\ \citenamefont {Martinez}}]{BPOTF}%
  \BibitemOpen
  \bibfield  {author} {\bibinfo {author} {\bibfnamefont {A.}~\bibnamefont {deMarti iOlius}}, \bibinfo {author} {\bibfnamefont {I.~E.}\ \bibnamefont {Martinez}}, \bibinfo {author} {\bibfnamefont {J.}~\bibnamefont {Roffe}},\ and\ \bibinfo {author} {\bibfnamefont {J.~E.}\ \bibnamefont {Martinez}},\ }\href {https://arxiv.org/abs/2409.01440} {\bibinfo {title} {An almost-linear time decoding algorithm for quantum ldpc codes under circuit-level noise}} (\bibinfo {year} {2024}),\ \Eprint {https://arxiv.org/abs/2409.01440} {arXiv:2409.01440 [quant-ph]} \BibitemShut {NoStop}%
\bibitem [{\citenamefont {Rakovszky}\ and\ \citenamefont {Khemani}(2023)}]{rakovszky2023ldpc}%
  \BibitemOpen
  \bibfield  {author} {\bibinfo {author} {\bibfnamefont {T.}~\bibnamefont {Rakovszky}}\ and\ \bibinfo {author} {\bibfnamefont {V.}~\bibnamefont {Khemani}},\ }\href {https://arxiv.org/abs/2310.16032} {\bibinfo {title} {The physics of (good) ldpc codes i. gauging and dualities}} (\bibinfo {year} {2023}),\ \Eprint {https://arxiv.org/abs/2310.16032} {arXiv:2310.16032 [quant-ph]} \BibitemShut {NoStop}%
\bibitem [{\citenamefont {Dinur}\ \emph {et~al.}(2022)\citenamefont {Dinur}, \citenamefont {Evra}, \citenamefont {Livne}, \citenamefont {Lubotzky},\ and\ \citenamefont {Mozes}}]{Dinur_2022_LTC}%
  \BibitemOpen
  \bibfield  {author} {\bibinfo {author} {\bibfnamefont {I.}~\bibnamefont {Dinur}}, \bibinfo {author} {\bibfnamefont {S.}~\bibnamefont {Evra}}, \bibinfo {author} {\bibfnamefont {R.}~\bibnamefont {Livne}}, \bibinfo {author} {\bibfnamefont {A.}~\bibnamefont {Lubotzky}},\ and\ \bibinfo {author} {\bibfnamefont {S.}~\bibnamefont {Mozes}},\ }\bibinfo {title} {Locally testable codes with constant rate, distance, and locality},\ in\ \href {https://doi.org/10.1145/3519935.3520024} {\emph {\bibinfo {booktitle} {Proceedings of the 54th Annual ACM SIGACT Symposium on Theory of Computing}}}\ (\bibinfo  {publisher} {Association for Computing Machinery},\ \bibinfo {address} {New York, NY, USA},\ \bibinfo {year} {2022})\ p.\ \bibinfo {pages} {357–374}\BibitemShut {NoStop}%
\bibitem [{\citenamefont {Gu}\ \emph {et~al.}(2024)\citenamefont {Gu}, \citenamefont {Tang}, \citenamefont {Caha}, \citenamefont {Choe}, \citenamefont {He},\ and\ \citenamefont {Kubica}}]{Gu_2024}%
  \BibitemOpen
  \bibfield  {author} {\bibinfo {author} {\bibfnamefont {S.}~\bibnamefont {Gu}}, \bibinfo {author} {\bibfnamefont {E.}~\bibnamefont {Tang}}, \bibinfo {author} {\bibfnamefont {L.}~\bibnamefont {Caha}}, \bibinfo {author} {\bibfnamefont {S.~H.}\ \bibnamefont {Choe}}, \bibinfo {author} {\bibfnamefont {Z.}~\bibnamefont {He}},\ and\ \bibinfo {author} {\bibfnamefont {A.}~\bibnamefont {Kubica}},\ }\href {http://dx.doi.org/10.1007/s00220-024-04951-6} {\bibfield  {journal} {\bibinfo  {journal} {Communications in Mathematical Physics}\ }\textbf {\bibinfo {volume} {405}},\ \bibinfo {pages} {85} (\bibinfo {year} {2024})}\BibitemShut {NoStop}%
\bibitem [{\citenamefont {Kovalev}\ and\ \citenamefont {Pryadko}(2013)}]{Kovalev_2013}%
  \BibitemOpen
  \bibfield  {author} {\bibinfo {author} {\bibfnamefont {A.~A.}\ \bibnamefont {Kovalev}}\ and\ \bibinfo {author} {\bibfnamefont {L.~P.}\ \bibnamefont {Pryadko}},\ }\href {https://doi.org/10.1103/PhysRevA.87.020304} {\bibfield  {journal} {\bibinfo  {journal} {Phys. Rev. A}\ }\textbf {\bibinfo {volume} {87}},\ \bibinfo {pages} {020304} (\bibinfo {year} {2013})}\BibitemShut {NoStop}%
\bibitem [{\citenamefont {Tan}\ and\ \citenamefont {Stambler}(2024)}]{tan2024hgp}%
  \BibitemOpen
  \bibfield  {author} {\bibinfo {author} {\bibfnamefont {S.~J.~S.}\ \bibnamefont {Tan}}\ and\ \bibinfo {author} {\bibfnamefont {L.}~\bibnamefont {Stambler}},\ }\href {https://arxiv.org/abs/2409.02193} {\bibinfo {title} {Effective distance of higher dimensional hgps and weight-reduced quantum ldpc codes}} (\bibinfo {year} {2024}),\ \Eprint {https://arxiv.org/abs/2409.02193} {arXiv:2409.02193 [quant-ph]} \BibitemShut {NoStop}%
\bibitem [{\citenamefont {Burton}\ and\ \citenamefont {Browne}(2022)}]{Burton_2022}%
  \BibitemOpen
  \bibfield  {author} {\bibinfo {author} {\bibfnamefont {S.}~\bibnamefont {Burton}}\ and\ \bibinfo {author} {\bibfnamefont {D.}~\bibnamefont {Browne}},\ }\href {https://doi.org/10.1109/TIT.2021.3131043} {\bibfield  {journal} {\bibinfo  {journal} {IEEE Transactions on Information Theory}\ }\textbf {\bibinfo {volume} {68}},\ \bibinfo {pages} {1772} (\bibinfo {year} {2022})}\BibitemShut {NoStop}%
\bibitem [{\citenamefont {Vasmer}\ and\ \citenamefont {Browne}(2019)}]{Vasmer_2019_3D}%
  \BibitemOpen
  \bibfield  {author} {\bibinfo {author} {\bibfnamefont {M.}~\bibnamefont {Vasmer}}\ and\ \bibinfo {author} {\bibfnamefont {D.~E.}\ \bibnamefont {Browne}},\ }\href {https://doi.org/10.1103/PhysRevA.100.012312} {\bibfield  {journal} {\bibinfo  {journal} {Phys. Rev. A}\ }\textbf {\bibinfo {volume} {100}},\ \bibinfo {pages} {012312} (\bibinfo {year} {2019})}\BibitemShut {NoStop}%
\end{thebibliography}%

\end{document}